\documentclass[letterpaper,11pt]{article}

\usepackage[showWADSdel]{optional}

\usepackage[margin=1in]{geometry}
\usepackage{setspace}
\usepackage{graphicx}
\usepackage{epsfig}

\usepackage{microtype} %
\usepackage{array} %

\usepackage{amssymb}
\usepackage{amsmath}
\usepackage{amsthm}
\usepackage{mathtools} %
\usepackage{stmaryrd} %
\usepackage{stackrel} %
\usepackage{latexsym} %

\usepackage{tikz}
\usetikzlibrary{arrows}
\usetikzlibrary{arrows.meta}
\usetikzlibrary{decorations.pathreplacing}
\usetikzlibrary{cd}
\usepackage[all]{xy} %

\usepackage[section]{algorithm}
\usepackage[noend]{algpseudocode}
\algnewcommand{\LineComment}[1]{\Statex\hspace{\algorithmicindent}\(\triangleright\) #1}
\algnewcommand\algorithmicforeach{\textbf{for each}}
\algdef{S}[FOR]{ForEach}[1]{\algorithmicforeach\ #1\ \algorithmicdo}
\usepackage{etoolbox} %

\usepackage[labelformat=parens]{subfig} %
\usepackage{caption}
\usepackage{picins}

\usepackage{makecell} %
\usepackage{booktabs} %
\usepackage{multirow} %

\usepackage[hidelinks]{hyperref} %
\usepackage[symbol]{footmisc}
\usepackage{makeidx} %
\usepackage{url} %
\usepackage{color}
\usepackage{xcolor}
\usepackage{xspace} %
\usepackage[normalem]{ulem} %
\usepackage{soul} %
\usepackage{ifpdf} %
\usepackage{extarrows}

\makeatletter
\expandafter\patchcmd\csname\string\algorithmic\endcsname{\itemsep\z@}{\itemsep=0.25ex}{}{}
\makeatother

\newcounter{usesmallsep}
\ifnum\the\value{usesmallsep}=1

    \newlength{\myitemsep}
    \newlength{\mytopsep}
    \usepackage{enumitem}
    \setlist[itemize]{leftmargin=\parindent,parsep=\parskip,
      listparindent=\parindent,itemsep=\myitemsep,topsep=\myitemsep}
    \setlist[enumerate]{leftmargin=\parindent,parsep=\parskip,
      listparindent=\parindent,itemsep=\myitemsep,,topsep=\myitemsep}
    \setlist[description]{font=\bfseries,leftmargin=\parindent,parsep=\parskip,
      listparindent=\parindent,itemsep=\myitemsep,topsep=\myitemsep}
    
    \newlength{\mypartitlesep}
    \usepackage[explicit]{titlesec}
    \titlespacing{\paragraph}{0pt}{\mypartitlesep}{\mypartitlesep}
    
    \newlength{\mythmsep}
    \newtheoremstyle{mythmstyle}
      {\mythmsep} %
      {\mythmsep} %
      {\itshape} %
      {} %
      {\bfseries} %
      {.} %
      {.5em} %
      {} %
    
    \newtheoremstyle{mydefstyle}
      {\mythmsep} %
      {\mythmsep} %
      {} %
      {} %
      {\bfseries} %
      {.} %
      {.5em} %
      {} %
    
    \theoremstyle{mythmstyle}
        \newtheorem{theorem}{Theorem}
        \newtheorem{proposition}[theorem]{Proposition}
        \newtheorem{lemma}[theorem]{Lemma}
        \newtheorem{corollary}[theorem]{Corollary}
        \newtheorem{fact}[theorem]{Fact}
        \newtheorem*{fact*}{Fact}
    \theoremstyle{mydefstyle}
        \newtheorem{definition}{Definition}
        \newtheorem{problem}{Problem}
        \newtheorem{assumption}{Assumption}
        \newtheorem{remark}{Remark}
        \newtheorem{algr}[algorithm]{Algorithm}

    \newenvironment{proof}
        {\vspace{-0.9em}\begin{proof}}
        {\end{proof}\vspace{-0.4em}}

\else

    \theoremstyle{plain}
        \newtheorem{theorem}{Theorem}
        \newtheorem{proposition}[theorem]{Proposition}

        \newtheorem{observation}[theorem]{Observation}
    \theoremstyle{definition}
        \newtheorem{definition}[theorem]{Definition}

        \newtheorem{remark}{Remark}
        \newtheorem*{remark*}{Remark}
        \newtheorem*{example*}{Example}
        \newtheorem{algr}{Algorithm}
        \newtheorem*{algr*}{Algorithm}
        
    \usepackage{enumitem}
    \setlist[itemize]{leftmargin=\parindent}
    \setlist[enumerate]{leftmargin=\parindent}
    \setlist[description]{font=\bfseries,leftmargin=\parindent}

\fi

\newcommand{\Hm}{\mathsf{H}}

\newcommand{\fsimp}[2]{\sigma_{#2}}

\newcommand{\Pers}{\mathsf{Pers}}

\newcommand{\lbarrowspace}{\;}

\let\leftrightarrowsp\lrarrowsp

\newcommand{\incto}{\hookrightarrow}
\newcommand{\inctosp}[1]{\xhookrightarrow{\lbarrowspace#1\lbarrowspace}}
\newcommand{\bakincto}{\hookleftarrow}
\newcommand{\bakinctosp}[1]{\xhookleftarrow{\lbarrowspace#1\lbarrowspace}}

\newcommand{\hatfsimp}{\hat{\sigma}}

\newcommand{\Ud}{{\mathcal U}}
\newcommand{\given}{\,|\,}
\newcommand{\Set}[1]{\{#1\}}

\newcommand{\MF}{\mathsf{MF}}

\newcommand{\idx}{\mathrm{idx}}

\let\emptyset\varnothing
\let\intersect\cap

\let\union\cup

\newcommand{\Fcal}{\mathcal{F}}

\newcommand{\Ical}{\mathcal{I}}

\newcommand{\Lcal}{\mathcal{L}}

\newcommand{\Ucal}{\mathcal{U}}

\newcommand{\aG}{\alpha}

\newcommand{\DG}{\Delta}
\newcommand{\eG}{\epsilon}
\newcommand{\gG}{\gamma}

\newcommand{\lG}{\lambda}
\newcommand{\LG}{\Lambda}

\newcommand{\sG}{\sigma}

\newcommand{\tG}{\tau}

\newcommand{\Dim}{p}

\newcommand{\birth}{b}
\newcommand{\death}{d}
\newcommand{\filtcnt}{m}
\newcommand{\simpcnt}{k}

\newcommand{\Fpo}{\mathbb{T}}

\newcommand{\defemph}[1]{\emph{#1}}

\theoremstyle{definition}

\begin{document}

\title{Revisiting Graph Persistence for Updates and Efficiency\thanks{This research is partially supported by NSF grant CCF 2049010.}}

\author{Tamal K. Dey\thanks{Department of Computer Science, Purdue University. \texttt{tamaldey@purdue.edu}}
\and Tao Hou\thanks{School of Computing, DePaul University. \texttt{thou1@depaul.edu}}
\and Salman Parsa\thanks{School of Computing, DePaul University. \texttt{s.parsa@depaul.edu}}
}

\date{}

\maketitle
\thispagestyle{empty}

\begin{abstract}
It is well known that ordinary persistence on graphs can be computed more
efficiently than the general persistence. Recently, it has  been shown
that zigzag persistence on graphs also exhibits similar behavior.
Motivated by these results, we revisit graph persistence 
and propose
efficient algorithms especially for 
local updates on filtrations, similar to what is done in ordinary persistence for computing the \emph{vineyard}. We show that, for
a filtration of length $m$,
(i)~switches
(transpositions) in ordinary graph persistence can be done in $O(\log m)$ time;
(ii)~zigzag persistence on graphs can be computed in $O(m\log m)$ time, which
improves a recent $O(m\log^4n)$ time algorithm assuming
$n$, the size of the union of all graphs in the filtration, satisfies
$n\in\Omega({m^\varepsilon})$ for any fixed $0<\varepsilon<1$; 
(iii) open-closed, closed-open,
and closed-closed bars
in dimension $0$ for graph zigzag persistence can be updated in
$O(\log m)$ time, whereas the open-open bars in dimension
$0$ and closed-closed bars in dimension $1$ can be done in $O(\sqrt{m}\,\log m)$ time.
\end{abstract}

\newpage
\setcounter{page}{1}

\section{Introduction}
Computing persistence for graphs has been a special focus
within topological data analysis (TDA)~\cite{DW22,edelsbrunner2010computational} because graphs are
abundant in applications and they admit more efficient
algorithms than general simplicial complexes. It is well known that the
persistence algorithm on a graph filtration with $m$ additions
can be implemented with 
a simple 
Union-Find
data structure in $O(m\,\alpha(m))$ time, where $\alpha(m)$
is the inverse Ackermann's function (see e.g.~\cite{DW22}).
On the other hand, the general-purpose persistence algorithm on a simplicial
filtration comprising $m$ simplices runs in $O(m^\omega)$ time~\cite{milosavljevic2011zigzag},
where $\omega< 2.373$ is the exponent for matrix multiplication.
In a similar vein, Yan et al.~\cite{yan2021link} have recently shown that 
extended persistence~\cite{cohen2009extending} for graphs can also be computed more efficiently
in $O(m^2)$ time.
The zigzag version~\cite{carlsson2010zigzag} of the problem also exhibits similar
behavior; see e.g.\ the survey~\cite{berkouk22}. Even though the general-purpose zigzag persistence
algorithm runs in $O(m^\omega)$ time on a zigzag filtration
with $m$ additions and deletions~\cite{carlsson2009zigzag-realvalue,DBLP:conf/esa/DeyH22,maria2014zigzag,milosavljevic2011zigzag},
a recent result in~\cite{dey2021computing} shows that graph zigzag persistence
can be computed in 
$O(m\log^4n)$ 
time using some
appropriate dynamic data structures~\cite{georgiadis2011data,holm2001poly}
($n$ is the size of the union of all graphs in the filtration).

Motivated by the above developments, we embark on
revisiting the graph persistence and find more efficient algorithms using appropriate
dynamic data structures, especially in the dynamic settings~\cite{cohen2006vines,dey2021updating}.
In a dynamic setting, the graph filtration changes, and we are
required to update the barcode (persistence diagram) accordingly. For general simplicial complexes as input, 
the \emph{vineyard} algorithm of~\cite{cohen2006vines} updates the barcode in $O(m)$ time
for a \emph{switch} of two  
consecutive simplices (originally called a \emph{transposition} in~\cite{cohen2006vines}). 
So, we ask
if a similar update can be done more efficiently for a graph filtration. 
We show that, using some appropriate
dynamic data structures, indeed we can execute such updates
more efficiently. Specifically, we show the following:
\begin{enumerate}
    \item In a standard (non-zigzag) graph filtration comprising $m$ additions,
    a switch can be implemented in $O(\log m)$ time with a preprocessing
    time of $O(m\log m)$. See Section~\ref{sec:std-switch}.
    As a subroutine of the update algorithm for switches on graph filtrations,
    we propose an update on the merge trees (termed as \emph{merge forest} in this paper)
    of the filtrations, 
    whose complexity is also $O(\log m)$.
    \item The barcode of a graph zigzag filtration comprising
    $m$ additions and deletions can be computed in $O(m\log m)$ time.
    Assuming $n\in \Omega(m^\varepsilon)$ for any fixed positive $\varepsilon<1$, where $n$ is the size of the union of all graphs in the filtration,
    this is an improvement over
    the $O(m\log^4 n)$ complexity of the algorithm in~\cite{dey2021computing}.
    See Section~\ref{sec:gzz-non-up}.
    Also, our current algorithm using Link-Cut tree~\cite{sleator1981data}
    is much easier to implement than the algorithm in~\cite{dey2021computing} using the Dynamic 
    Minimum Spanning Forest~\cite{holm2001poly}.
    \item 
    For switches~\cite{dey2021updating} on graph zigzag persistence, 
    the \emph{closed-closed} intervals in dimension 0
    can be maintained in $O(1)$ time;
    the \textit{closed-open} and \textit{open-closed} intervals, which appear only in dimension $0$,
    can be maintained in $O(\log m)$ time;
    the \emph{open-open} intervals in dimension 0 and 
    \emph{closed-closed} intervals in dimension 1
    can be maintained in $O(\sqrt{m}\,\log m)$ time. 
    All these can be done with an $O(m^{1.5}\log m)$ preprocessing
    time.
    See Section~\ref{sec:zz-switch}.
\end{enumerate}

\section{Preliminaries}

\paragraph{Graph zigzag persistence.}

A {\it graph zigzag filtration} 
is a sequence of graphs
\begin{equation}\label{eqn:prelim-zz}
\Fcal: G_0 \leftrightarrow G_1 \leftrightarrow 
\cdots \leftrightarrow G_\filtcnt,
\end{equation}
in which each
$G_i\leftrightarrow G_{i+1}$ is either a forward inclusion $G_i\incto G_{i+1}$
or a backward inclusion $G_i\bakincto G_{i+1}$.
For computation,
we only consider \emph{simplex-wise} filtrations 
starting and ending with \emph{empty} graphs
in this paper, i.e.,
$G_0=G_\filtcnt=\emptyset$ and
each inclusion $G_i\leftrightarrow G_{i+1}$ is an addition or deletion
of a single vertex or edge (both called a \emph{simplex}).
Such an inclusion is sometimes denoted as $G_i\leftrightarrowsp{\sG} G_{i+1}$
with $\sG$ indicating the vertex or edge being added or deleted.
The $\Dim$-th homology functor ($p=0,1$) applied on $\Fcal$
induces a {\it zigzag module}:
\begin{equation*}
\Hm_\Dim(\Fcal): 
\Hm_\Dim(G_0) 
\leftrightarrow
\Hm_\Dim(G_1) 
\leftrightarrow
\cdots 
\leftrightarrow
\Hm_\Dim(G_\filtcnt), 
\end{equation*}
in which
each $\Hm_\Dim(G_i)\leftrightarrow \Hm_\Dim(G_{i+1})$
is a linear map induced by inclusion.
It is known~\cite{carlsson2010zigzag,Gabriel72} that
$\Hm_\Dim(\Fcal)$ has a decomposition of the form
$\Hm_\Dim(\Fcal)\simeq\bigoplus_{k\in\LG}\Ical^{[\birth_k,\death_k]}$,
in which each $\Ical^{[\birth_k,\death_k]}$
is an
{\it interval module} over the interval $[\birth_k,\death_k]$.
The multiset of intervals
$\Pers_\Dim(\Fcal):=\Set{[\birth_k,\death_k]\given k\in\LG}$
is an invariant of $\Fcal$
and is called the {\it $\Dim$-th barcode} of $\Fcal$.
Each interval in $\Pers_\Dim(\Fcal)$ is called a {\it $\Dim$-th persistence interval}
and is also said to be in dimension $\Dim$.
Frequently in this paper, we consider the barcode of $\Fcal$ in all dimensions
$\Pers_*(\Fcal):=\bigsqcup_{\Dim= 0,1}\Pers_\Dim(\Fcal)$.

\paragraph{Standard persistence and simplex pairing.}
If all inclusions in Equation~(\ref{eqn:prelim-zz})
are forward, we have a \emph{standard} (\emph{non-zigzag}) graph filtration.
We also only consider standard graph filtrations that are simplex-wise
and start with empty graphs. Let $\Fcal$ be such a filtration.
It is well-known~\cite{edelsbrunner2000topological} that $\Pers_*(\Fcal)$
is generated from a \emph{pairing} of simplices in $\Fcal$ s.t.\ for each pair $(\sG,\tG)$
generating a $[b,d)\in\Pers_*(\Fcal)$,
the simplex $\sG$ creating $[b,d)$ is called \emph{positive} 
and $\tG$ destroying $[b,d)$ is called \emph{negative}.
Notice that $d$ may equal $\infty$ for a $[b,d)\in\Pers_*(\Fcal)$,
in which case $[b,d)$ is generated by an \emph{unpaired} positive simplex.
For a simplex $\sG$ added from $G_{i}$ to $G_{i+1}$ in $\Fcal$, we let its \emph{index}
be $i$ and denote it as $\idx_\Fcal(\sG):=i$.
For another simplex $\tG$ added in $\Fcal$,
if $\idx_\Fcal(\sG)<\idx_\Fcal(\tG)$, 
we say that $\sG$ is \emph{older} than $\tG$
and $\tG$ is \emph{younger} than $\sG$.

\paragraph{Merge forest.}
{Merge forests} (more commonly called {merge trees})
encode 
the 
evolution of connected components in a standard graph filtration~\cite{edelsbrunner2010computational,Parsa13}.
We adopt merge forests as central constructs in our update algorithm for standard graph persistence
(Algorithm~\ref{alg:std-switch-abs}).
We rephrase its definition
below:

\begin{definition}[Merge forest]\label{dfn:MF}
For a simplex-wise standard graph filtration
\[\Fcal:\emptyset= G_0 \inctosp{\sG_0}
G_1 \inctosp{\sG_1}
\cdots
\cdots \inctosp{\sG_{\filtcnt-1}} G_\filtcnt,\]
its \defemph{merge forest}
$\MF(\Fcal)$
is a forest (acyclic undirected graph) where the leaves correspond to vertices in $\Fcal$
and the internal nodes correspond to negative edges  in $\Fcal$.
Moreover, each node in $\MF(\Fcal)$ is associated with a \defemph{level} which is the   index
of its corresponding simplex in $\Fcal$.
Let $\MF^i(\Fcal)$ be the subgraph of $\MF(\Fcal)$ induced by nodes
at levels less than $i$.
Notice that trees in $\MF^i(\Fcal)$ bijectively correspond
to connected components in $G_{i}$.
We then constructively define $\MF^{i+1}(\Fcal)$ from $\MF^{i}(\Fcal)$,
starting with $\MF^{0}(\Fcal)=\emptyset$
and ending with $\MF^{\filtcnt}(\Fcal)=\MF(\Fcal)$.
Specifically, for each $i=0,1,\ldots,\filtcnt-1$, do the following:
\begin{description}
    \item[$\sG_{i}$ is a vertex:]
    $\MF^{i+1}(\Fcal)$ equals $\MF^{i}(\Fcal)$ union an isolated leaf at level $i$
    corresponding to $\sG_i$.
    \item[$\sG_{i}$ is a positive edge:]
    Set $\MF^{i+1}(\Fcal)=\MF^{i}(\Fcal)$.
    \item[$\sG_{i}$ is a negative edge:]
    Let $\sG_i=(u,v)$. Since $u$ and $v$ are in different connected components $C_1$ and $C_2$ in $G_i$,  
    let $T_1,T_2$ be the trees in $\MF^{i}(\Fcal)$
    corresponding to $C_1,C_2$ respectively.
    To form $\MF^{i+1}(\Fcal)$, we add an internal node at level $i$ (corresponding to $\sG_i$) to $\MF^{i}(\Fcal)$ 
    whose children are the roots of $T_1$ and $T_2$.
\end{description}
\end{definition}

In this paper, we do not differentiate a vertex or edge in $\Fcal$
and its corresponding node in $\MF(\Fcal)$.

\section{Updating standard persistence on graphs}
\label{sec:std-switch}

The switch operation originally proposed in~\cite{cohen2006vines} for
general filtrations looks as follows on standard graph filtrations:
\begin{equation}\label{eqn:std-switch}
\begin{tikzpicture}[baseline=(current  bounding  box.center)]
\tikzstyle{every node}=[minimum width=24em]
\node (a) at (0,0) {$\Fcal: \emptyset= G_0 \incto
\cdots
\incto
G_{i-1}\inctosp{\sG} 
G_i 
\inctosp{\tG} G_{i+1}
\incto
\cdots \incto G_\filtcnt$}; 
\node (b) at (0,-0.6){$\Fcal': \emptyset= G_0 \incto
\cdots
\incto 
G_{i-1}\inctosp{\tG} 
G'_i 
\inctosp{\sG} G_{i+1}
\incto
\cdots \incto G_\filtcnt$};
\path[->] (a.0) edge [bend left=90,looseness=1.5,arrows={-latex},dashed] (b.0);
\end{tikzpicture}
\end{equation}
In the above operation, 
the addition of two simplices 
$\sG$ and $\tG$
are switched from $\Fcal$ to $\Fcal'$.
We also require that $\sG\not\subseteq\tG$~\cite{cohen2006vines} ($\sG$ is not a vertex of the edge $
\tG$)
because otherwise $G'_i$ is not a valid graph.

For a better presentation, 
we provide the idea at a high level for the updates
in Algorithm~\ref{alg:std-switch-abs};
the full details are presented in 
Algorithm~\ref{alg:std-switch-full} in Section~\ref{sec:alg-std-switch-full}.

We also notice the following fact about the change on pairing 
caused by 
the switch in Equation~(\ref{eqn:std-switch})
when $\sG,\tG$ are both positive or both negative.
Let $\sG$ be paired with $\sG'$ and 
$\tG$ be paired with $\tG'$ in $\Fcal$. (If $\sG$ or $\tG$ are \emph{unpaired},
then let $\sG'$ or $\tG'$ be $null$.)
By the update algorithm for general complexes~\cite{cohen2006vines},
either (i)~the pairing for $\Fcal$ and $\Fcal'$ stays the same,
or (ii) the only difference on the pairing is that
$\sG$ is paired with $\tG'$ and
$\tG$ is paired with $\sG'$ in $\Fcal'$.

\begin{algr}[Update for switch on standard graph filtrations]\label{alg:std-switch-abs}
\begin{itemize}\item[]\end{itemize}
\noindent For the switch operation in Equation~(\ref{eqn:std-switch}),
the algorithm maintains a merge forest $\Fpo$
(which initially represents $\MF(\Fcal)$)
and a pairing of simplices $\Pi$
(which initially corresponds to $\Fcal$).
The algorithm makes changes
to $\Fpo$ and $\Pi$
so that they correspond to $\Fcal'$
after the processing.
For an overview of the algorithm,
we describe the processing only for
the  cases (or sub-cases)
where we need to make changes to $\Fpo$ or $\Pi$:
\begin{description}
\parpic[r]{%
  \begin{minipage}{0.27\linewidth}
  \includegraphics[width=\linewidth]{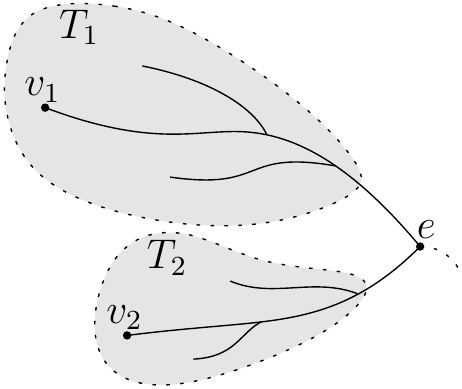}
  \captionof{figure}{}
  \label{fig:v-v-sw}
  \end{minipage}
  }
\item[A. The switch is a vertex-vertex switch:]
First 
let 
    $v_1:=\sG$ and $v_2:=\tG$.
    As illustrated in  Figure~\ref{fig:v-v-sw},
the only situation where the  pairing $\Pi$ changes in this case is that  $v_1,v_2$ are in
  the same tree in $\Fpo$
  and are both unpaired when $e$ 
  is added in $\Fcal$,
  where $e$ is 
    the edge corresponding to 
    the \emph{nearest common ancestor}
    of $v_1,v_2$ in $\Fpo$. 
  In this case, $v_1$,$v_2$ are leaves at the \emph{lowest} levels in the subtrees $T_1,T_2$ respectively (see Figure~\ref{fig:v-v-sw}),
  so that $v_2$ is paired with $e$ and $v_1$ is the \emph{representative} (the only unpaired vertex)
  in the merged connected component due to the addition of $e$.
  After the switch, $v_1$ is paired with $e$ due to being younger and $v_2$ becomes the representative
  of the merged component.
   Notice that the structure of $\Fpo$  stays the same.

\end{description}
\begin{description}
\item[B. The switch is an edge-edge switch:]
Let 
    $e_1:=\sG$ and $e_2:=\tG$.
  We  have  the following sub-cases:
  \medskip
   \begin{description}
       \item[B.1. $e_1$ is negative and $e_2$ is positive:]
       We need to make changes 
       when
       $e_1$ is in a 1-cycle in $G_{i+1}$ (see Figure~\ref{fig:e1-e2-same-conns}),
       which is equivalent to saying that
       $e_1,e_2$ connect to the same two connected components in $G_{i-1}$.
       In this case, $e_1$ becomes positive and 
       $e_2$ becomes negative after the switch, for which we pair $e_2$ with
       the vertex that $e_1$ previously pairs with.
       The node in $\Fpo$ corresponding to $e_1$ should now correspond to $e_2$
       after the switch.
     
    \begin{figure}[!t]
  \centering
  \includegraphics[width=0.6\linewidth]{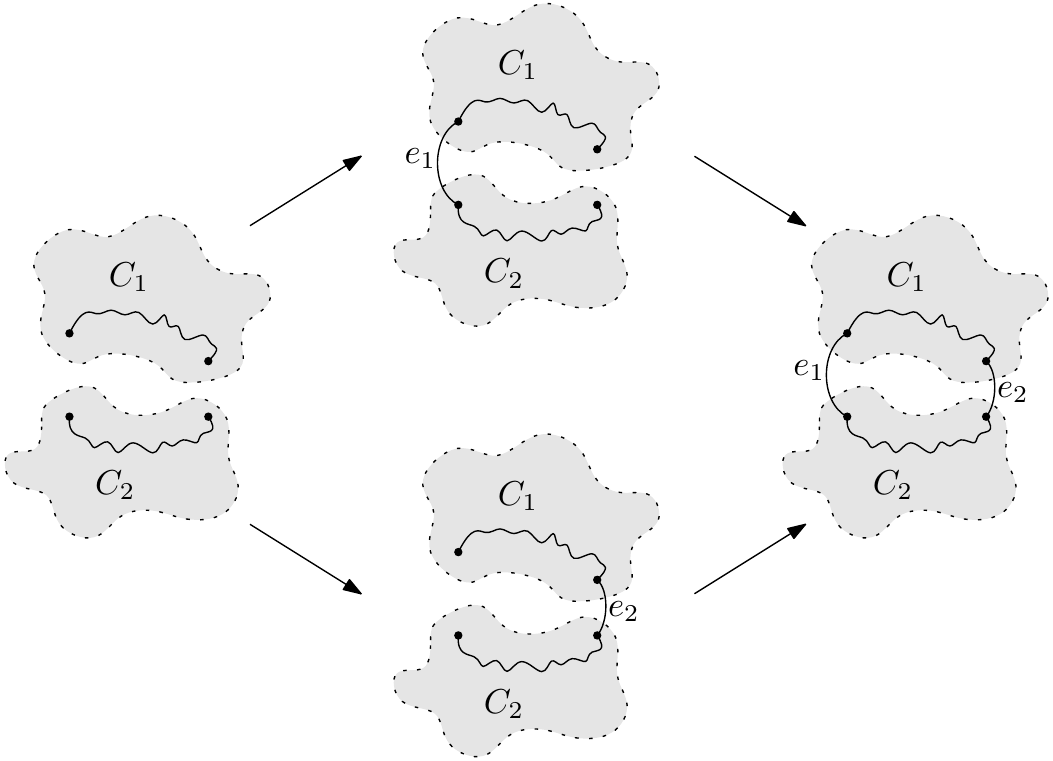}
  \caption{The edges $e_1,e_2$ connect to the same two connected components causing
  the change 
  in an edge-edge switch where $e_1$ is negative and $e_2$ is positive.}
  \label{fig:e1-e2-same-conns}
\end{figure}

  \begin{figure}[!tbh]
  \centering
   \subfloat[]{
        \centering
        \includegraphics[width=0.17\textwidth]{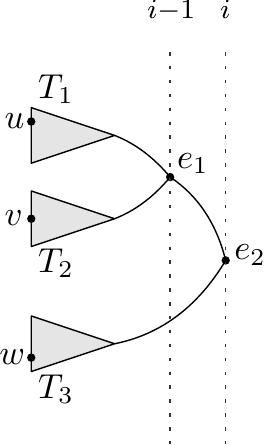} 
        \label{fig:e-e-sw-a}
    }\hspace{5em}
    \subfloat[]{
        \centering
        \includegraphics[width=0.17\textwidth]{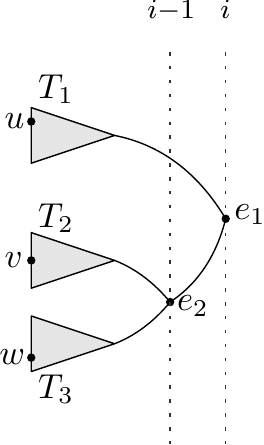} 
        \label{fig:e-e-sw-b}
    }\hspace{5em}
    \subfloat[]{
        \centering
        \includegraphics[width=0.17\textwidth]{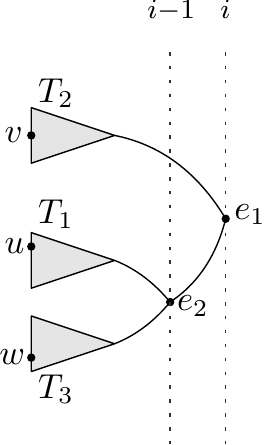} 
        \label{fig:e-e-sw-c}
    }
  \caption{(a) The relevant parts of $\Fpo$ when switching two negative edges in Algorithm~\ref{alg:std-switch-abs} for which the structure of $\Fpo$ changes.
  (b) The changed structure of $\Fpo$ after the switch
  corresponding to the  connecting configuration in Figure~\ref{fig:C1-C2-C3-1}.
  (c) The changed structure of $\Fpo$ after the switch
  corresponding to the  connecting configuration in Figure~\ref{fig:C1-C2-C3-2}.}
  \label{fig:e-e-sw}
\end{figure}

      \medskip
      \item[B.2. $e_1$ and $e_2$ are both  negative:]
    We need to make changes when
    the corresponding node of $e_1$ is a child of the corresponding node of $e_2$
           in $\Fpo$ (see Figure~\ref{fig:e-e-sw}a).
        To further illustrate the situation,
           let $T_1,T_2$ be the subtrees rooted at the two children of $e_1$ in $\Fpo$,
           and let $T_3$ be the subtree rooted at the other child of $e_2$ 
           that is not $e_1$ (as in Figure~\ref{fig:e-e-sw}a).
           Moreover,
           let $u,v,w$ be the leaves at the lowest levels in $T_1,T_2,T_3$ respectively. 
           Without loss of generality, assume that $\idx_\Fcal(v)<\idx_\Fcal(u)$.
           Since $T_1,T_2,T_3$ can be considered as trees in $\MF^{i-1}(\Fcal)$,
           let $C_1,C_2,C_3$ be the connected components of $G_{i-1}$
           corresponding to $T_1,T_2,T_3$ respectively (see Definition~\ref{dfn:MF}).

\begin{figure}[!t]
    \centering
    \subfloat[]{
        \centering
        \includegraphics[width=0.33\textwidth]{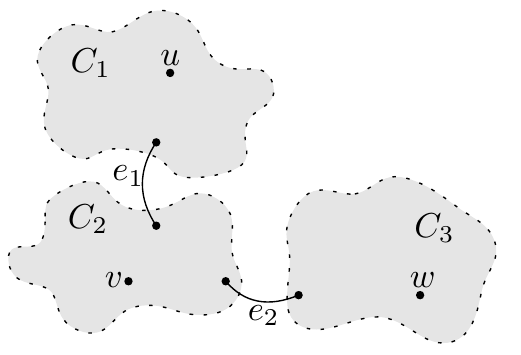} 
        \label{fig:C1-C2-C3-1}
    }
    \hspace{3em}
    \subfloat[]{
        \includegraphics[width=0.33\textwidth]{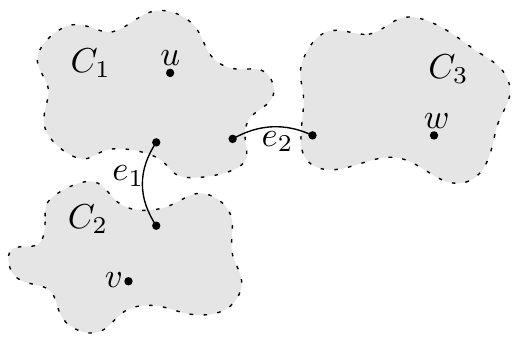} 
        \label{fig:C1-C2-C3-2}
    }
    \caption{Two different ways in which the edges connect the three components $C_1,C_2,C_3$ in $G_{i+1}$
    when switching two negative edges $e_1,e_2$.
    While $e_1$ always connects $C_1$ and $C_2$, $e_2$ could either directly connect $C_2,C_3$ (a) 
    or $C_1,C_3$ (b).}
    \label{fig:C1-C2-C3} 
\end{figure}

           We have that $C_1,C_2,C_3$ are connected by $e_1,e_2$ in $G_{i+1}$
           in the two different ways illustrated in Figure~\ref{fig:C1-C2-C3}.
           For the two different connecting configurations,
           the structure of $\Fpo$ after the switch is different,
           which is  shown in Figure~\ref{fig:e-e-sw}b
           and~\ref{fig:e-e-sw}c.
        Furthermore, if 
               $\idx_\Fcal(w)<\idx_\Fcal(u)$
               and $e_2$ directly connects $C_1,C_3$ as in Figure~\ref{fig:C1-C2-C3-2}, 
               then we swap the paired vertices of $e_1,e_2$
               in $\Pi$.
               (See Section~\ref{sec:alg-std-switch-full} for further details and  justifications.)
           
   \end{description}
\end{description}

\noindent In all cases,
the algorithm also updates the levels of the leaves in $\Fpo$ corresponding to $\sG$ and $\tG$
(if such leaves exist) due to the change of indices for the vertices.
Notice that
the positivity/negativity of simplices can be easily read off from the 
simplex pairing $\Pi$.
\end{algr}

\paragraph{Data structure for merge forests.}
We use the \emph{Depth First Tour Tree} (DFT-Tree) proposed by Farina and Laura~\cite{farina2015dynamic}
to implement the merge forest $\Fpo$, which supports the following operations:

\begin{itemize}
    \item $\textsc{root}(v)$: Returns the root of the tree containing node $v$.
    \item $\textsc{cut}(v)$: Deletes the edge connecting node $v$ to its parent.
    \item $\textsc{link}(u,v)$: Makes the root of the tree containing node $v$ be a child of node $u$.
    \item $\textsc{nca}(u,v)$: Returns the nearest common ancestor of two nodes $u,v$
    in the same tree.
    \item $\textsc{change-val}(v,x)$: Assigns the value associated to a leaf $v$ to be $x$.
    \item $\textsc{subtree-min}(v)$: Returns the leaf with the minimum associated value in the subtree rooted at $v$.
\end{itemize}

Let $N$ be the number of nodes.
All above operations in DFT-Tree
take $O(\log N)$ time.
Among the operations, $\textsc{root}$ is used to determine whether two nodes in $\Fpo$ are from the same tree;
$\textsc{cut}$ and $\textsc{link}$ are used to make the structural changes as in Figure~\ref{fig:e-e-sw};
$\textsc{change-val}$ is used to record (and update during switches) 
the levels of leaves;
$\textsc{subtree-min}$ is used to return the leaf at the lowest level in a subtree.

\paragraph{Detecting cycles.}

Algorithm~\ref{alg:std-switch-abs}  needs to check whether 
the edge $e_1$ resides in a 1-cycle in $G_{i+1}$ (see the edge-edge switch).
We show how to check this by dynamically maintaining 
a minimum spanning forest (MSF) for $G:=G_m$ using the Link-Cut tree~\cite{sleator1981data}
in $O(\log m)$ time per switch.
First,
let the weight of each edge $e$ of $G$
be the index of $e$ in $\Fcal$,
i.e., if $e$ is added to a graph $G_j$ in $\Fcal$,
then $w(e)=j$.
Moreover,
for a path 
in a graph,
define its \emph{bottleneck weight} as the maximum weight of edges on the path.
By representing the MSF of $G$ as a Link-Cut tree,
we use the following procedure to
check whether 
$e_1$ resides in a 1-cycle in $G_{i+1}$
when switching two edges $e_1,e_2$
with $e_1$ negative and $e_2$ positive:
\begin{enumerate}
    \item 
    Let $e_2=(u,v)$.
    Retrieve the bottleneck weight $w_\infty$ of the path connecting
  $u,v$ in the MSF of $G$.
  This can be done in $O(\log m)$ time
  using the Link-Cut tree.
  \item
  If $w_\infty$ equals $i-1$ (which is the weight of $e_1$), 
  then $e_1$ is in a 1-cycle in $G_{i+1}$;
  otherwise, $e_1$ is not in a 1-cycle in $G_{i+1}$. 
\end{enumerate}

To see the correctness of the above procedure,
  first consider the case that $e_1$ is in a 1-cycle in $G_{i+1}$.
  We must have that $e_1,e_2$ connect to the same two components
  as shown in Figure~\ref{fig:e1-e2-same-conns}
  (see the proof of Proposition~\ref{prop:e-e-switch-pos-neg}).
  Consider running the Kruskal's algorithm on the sequence of edges in $\Fcal$:
  It must be true that 
  $u,v$ first become connected 
  after $e_1$ is added to the partial MSF maintained by the Kruskal's algorithm.
  Therefore, the bottleneck weight of the path connecting $u,v$ 
  in the MSF of $G$ 
  equals $i-1$ which is the weight of $e_1$.
  Now consider the case that $e_1$ is not in a 1-cycle in $G_{i+1}$.
  The addition of $e_2$ creates a 1-cycle in $G'_i$ because $e_2$
  is positive in $\Fcal'$.
  So we have that $u,v$
  are already connected in $G_{i-1}$.
  Consider running the Kruskal's algorithm on the sequence of edges in $\Fcal$.
  Before adding $e_1$ to the partial MSF in Kruskal's algorithm,
  $u,v$ are already connected by a path in the partial MSF.
  Therefore, the bottleneck weight of the path connecting $u,v$ 
  in the MSF of $G$ 
  must be less than $i-1$.

For the switch from $\Fcal$ to $\Fcal'$,
we also need to update the MSF of $G$
due to the weight change on certain edges.
We first notice that switching a vertex and an edge
does not change the structure of the MSF
and the update is easy.
For a switch of two edges $e_1,e_2$, we do the following 
in the different cases:
\begin{description}
  \item[$e_1,e_2$ are both positive:]
  After the switch, $e_1$ and $e_2$ are still positive.
  Since only negative edges appear in the MSF of $G$ (see, e.g., the Kruskal's algorithm), 
  they are not in the MSF of $G$ 
  before and after the switch
  and therefore
  the MSF stays the same.
  \item[$e_1$ is positive and $e_2$ is negative:]
  Based on Proposition~\ref{prop:e-e-switch-pos-neg}
  in  Section~\ref{sec:pf-e-e-switch}, 
  $e_1$ stays positive and $e_2$ stays negative after the switch.
  Therefore, the fact that
  $e_1$ is not in the MSF and $e_2$ is in the MSF
  does not change before and after the switch.
  So we have that the structure of the MSF stays the same.
  Notice that due to the switch, the weight of $e_1$ in the MSF 
  changes, whose update can be done by
  cutting $e_1$ first and then re-linking it.
  \item[$e_1,e_2$ are both negative:]
  After the switch, $e_1$ and $e_2$ are still negative.
  The structure of the MSF stays the same because 
   $e_1,e_2$ are both in the MSF before and after the switch.
   Notice that we also need to update the weights for the two edges
in the Link-Cut tree as done previously.
  \item[$e_1$ is negative and $e_2$ is positive:]
  As in Algorithm~\ref{alg:std-switch-abs},
  we have the following subcases:
  \begin{description}
  \item[$e_1$ is in a 1-cycle in $G_{i+1}$:]
  In this case, $e_1$ becomes positive and $e_2$ becomes negative
  after the switch (see Proposition~\ref{prop:e-e-switch-pos-neg}).
  Therefore, after the switch, we should delete $e_1$
  from the MSF
  and add $e_2$ to the MSF.
  This can be done by cutting $e_1$ and then linking $e_2$ in the Link-Cut tree.
  \item[$e_1$ is not in a 1-cycle in $G_{i+1}$:]
  In this case, $e_1$ stays negative and $e_2$ stays positive
  after the switch (see Proposition~\ref{prop:e-e-switch-pos-neg}).
  The structure of the MSF does not change and we only 
  need to update the weight of $e_1$ as done previously.
  \end{description}

\end{description}

\begin{remark}
Updating the MSF for $G$  
can be done in $O(\log^4 m)$ amortized time
by directly utilizing the Dynamic-MSF data structure
proposed by Holm et al.~\cite{holm2001poly}.
However, according to the  cases described above,
our update on the MSF is not `fully-dynamic'
because it only entails simple operations such as 
changing the weights by one unit for one or two edges and possibly switching two edges
whose weights are consecutively ordered.
Therefore, we could achieve an $O(\log m)$ update
in the worst case simply using the Link-Cut tree~\cite{sleator1981data}.
\end{remark}

\paragraph{Determining different connecting configurations.}
We also use the MSF of $G$ (represented by Link-Cut trees) 
to determine the different connecting configurations
in Figure~\ref{fig:C1-C2-C3} for the relevant case in Algorithm~\ref{alg:std-switch-abs}.
Specifically, if $u,w$  in Figure~\ref{fig:C1-C2-C3}
are already connected in $G'_i$, then the configuration in
Figure~\ref{fig:C1-C2-C3-2} applies;
otherwise, the the configuration in
Figure~\ref{fig:C1-C2-C3-1} applies.

To see how to use the MSF to determine the connectivity of vertices
in a filtration, 
we first rephrase Proposition 21 in the full version\footnote{\url{https://arxiv.org/pdf/2103.07353.pdf}} of~\cite{dey2021computing} as follows: 
\begin{proposition}\label{prop:conn-msf}
For a graph filtration $\Lcal:H_1\incto H_2\incto \cdots \incto H_s$, 
assign a weight $\idx_\Lcal(e)$ to each edge $e$ in $H:=H_s$.
For two vertices $x,y$ that are connected in $H$, 
let $\pi$ be the unique path connecting $x,y$ in the unique minimum spanning forest
of $H$
and let $w_\infty$ be the bottleneck weight of $\pi$.
Then,
the index of the first graph in $\Lcal$ where $x,y$
are connected equals $w_\infty+1$.
\end{proposition}

For two vertices $x,y$ in a graph filtration $\Lcal$, 
checking whether $x,y$ are connected in a graph $H_j$ in $\Lcal$
boils down to finding the first graph $H_{j_*}$ in $\Lcal$ where $x,y$ are connected,
and then checking whether $H_j$ is before or after $H_{j_*}$.
This can be done in $O(\log m)$ time by finding the bottleneck weight 
of the path connecting $x,y$ in the MSF of $H$
based on Proposition~\ref{prop:conn-msf}.

\paragraph{Time complexity.}
The costliest
steps of Algorithm~\ref{alg:std-switch-abs}
are the operations on the DFT-Tree and Link-Cut tree.
Hence, Algorithm~\ref{alg:std-switch-abs} takes $O(\log m)$ time in the worst case.

\paragraph{Preprocessing.}
In order to perform a sequence of updates on a given filtration,
we also need to construct the data structures maintained by 
Algorithm~\ref{alg:std-switch-abs}
for the initial filtration.
Constructing $\Pi$ is nothing but computing standard persistence pairs on graphs,
which takes $O(m\,\aG(m))$ time using Union-Find.
Constructing $\Fpo$ entails continuously performing the $\textsc{link}$ operations
on the DFT-Tree,
which takes $O(m\log m)$ time.
For constructing the MSF, 
we simply run the Kruskal's algorithm,
which can be done in $O(m\log m)$ time.
Hence, the preprocessing takes $O(m\log m)$ time.

\subsection{Full details and justifications for Algorithm~\ref{alg:std-switch-abs}}
\label{sec:alg-std-switch-full}
We present the full details of Algorithm~\ref{alg:std-switch-abs} as follows:

\begin{algr}[Full details of Algorithm~\ref{alg:std-switch-abs}]\label{alg:std-switch-full}
For the switch operation in Equation~(\ref{eqn:std-switch}),
the algorithm maintains a merge forest $\Fpo$
(which initially represents $\MF(\Fcal)$)
and a pairing of simplices $\Pi$
(which initially corresponds to $\Fcal$).
The algorithm makes changes
to $\Fpo$ and $\Pi$
so that they correspond to $\Fcal'$
after the processing.
Specifically, 
it does the following according to different cases:

\smallskip\noindent
\underline{If the switch is a vertex-edge switch or an edge-vertex switch},
then do nothing.

\smallskip
\noindent \underline{If the switch is a vertex-vertex switch}, let 
    $v_1:=\sG$, $v_2:=\tG$.
    If $v_1,v_2$ are in the same tree in $\Fpo$,
    then do the following:
        \begin{itemize}
            \item
            Find the \emph{nearest common ancestor} $x$ of $v_1,v_2$ in $\Fpo$
            and let $e$ be the edge corresponding to $x$.
            If both of the following are true:
            \begin{itemize}
                \item $v_1$ is unpaired in $\Pi$ or $v_1$ is paired with an $e_1$ in $\Pi$ s.t.\ $\idx_\Fcal(e_1)\geq \idx_\Fcal(e)$
                \item $v_2$ is unpaired in $\Pi$ or $v_2$ is paired with an $e_2$ in $\Pi$ s.t.\ $\idx_\Fcal(e_2)\geq \idx_\Fcal(e)$
            \end{itemize}
            then  swap the paired simplices of $v_1,v_2$ in $\Pi$. Notice that 
            $v_1$ or $v_2$ may be unpaired in $\Pi$, e.g., 
            we could have that $v_1$ is paired with $e_1$ and $v_2$ is unpaired,
            in which case $v_2$ becomes paired with $e_1$
            and $v_1$ becomes unpaired after the swap.
            
        \end{itemize}

\noindent \underline{If the switch is an edge-edge switch}, let 
    $e_1:=\sG$, $e_2:=\tG$.
  We have the following sub-cases:
   \begin{description}
       \item[$e_1$ and $e_2$ are both  positive:]
       Do nothing.
       \item[$e_1$ is positive and $e_2$ is negative:]
       Do nothing.
       \item[$e_1$ is negative and $e_2$ is positive:]
       If $e_1$ is in a 1-cycle in $G_{i+1}$,
       then: let the node corresponding to $e_1$ in $\Fpo$
       now correspond to $e_2$;
       let the vertex paired with $e_1$ in $\Pi$ now be paired with $e_2$;
       let $e_1$ be unpaired in $\Pi$. 
       \item[$e_1$ and $e_2$ are both  negative:]
    If (the corresponding node of) $e_1$ is a child of (the corresponding node of) $e_2$
           in $\Fpo$, then do the following:
       \begin{itemize}
           \item
           Let $T_1,T_2$ be the subtrees rooted at the two children of $e_1$ in $\Fpo$.
           Furthermore, let $c\neq e_1$ be the other child of $e_2$,
           and let $T_3$ be the subtree rooted at $c$ (see Figure~\ref{fig:e-e-sw}a). 
           Since $T_1,T_2,T_3$ can be considered as trees in $\MF^{i-1}(\Fcal)$,
           let $C_1,C_2,C_3$ be the connected components of $G_{i-1}$
           corresponding to $T_1,T_2,T_3$ respectively.
           We have that $C_1,C_2,C_3$ are connected by $e_1,e_2$ in $G_{i+1}$
           in the two different ways illustrated in Figure~\ref{fig:C1-C2-C3}.
           
           Let $u,v,w$ be the leaves at the lowest (smallest) levels in $T_1,T_2,T_3$ respectively. 
           WLOG, assume that $\idx_\Fcal(v)<\idx_\Fcal(u)$.
           We further have the following cases:
           \begin{itemize}
               \item
               If $e_2$ directly connects $C_2,C_3$ as in Figure~\ref{fig:C1-C2-C3-1},
               then let the roots of $T_2,T_3$ be the children of $e_2$,
               and let $e_2$ and the root of $T_1$ be the children of $e_1$ in $\Fpo$ (see Figure~\ref{fig:e-e-sw}b).
               \item
               If $e_2$ directly connects $C_1,C_3$ as in Figure~\ref{fig:C1-C2-C3-2},
               then let the roots of $T_1,T_3$ be the children of $e_2$,
               and let $e_2$ and the root of $T_2$ be the children of $e_1$ in $\Fpo$ (see Figure~\ref{fig:e-e-sw}c).
               Moreover,
               if 
               $\idx_\Fcal(w)<\idx_\Fcal(u)$, then  swap the paired vertices of $e_1,e_2$
               in $\Pi$.
           \end{itemize}
       \end{itemize}
   \end{description}

\noindent In all cases,
the algorithm also updates the levels of the leaves in $\Fpo$ corresponding to $\sG$ and $\tG$
(if such leaves exist) due to the change of indices for the vertices.
Notice that
the positivity/negativity of simplices can be easily read off from the 
simplex pairing $\Pi$.
\end{algr}

In the rest of the section, we justify the correctness of Algorithm~\ref{alg:std-switch-abs}
for all cases. 

\subsubsection{Justification for vertex-vertex switch} 

\begin{proposition}\label{prop:std-vv-pi}
For the switch operation in Equation~(\ref{eqn:std-switch})
where $v_1:=\sG$ and $v_2:=\tG$ are vertices ($v_1,v_2$ are thus both positive),
the pairings for $\Fcal$ and $\Fcal'$  change
if and only if the following two 
conditions hold:
\begin{enumerate}
    \item $v_1$ and $v_2$ are in the same tree in $\MF(\Fcal)$;
    \item $v_1$ and $v_2$ are both  unpaired when 
    $e$ is added in $\Fcal$, where $e$ is 
    the edge corresponding to 
    the nearest common ancestor $x$ of $v_1,v_2$ in $\MF(\Fcal)$.
\end{enumerate}
\end{proposition}
\begin{proof}
Suppose that the two conditions hold.
Let $j=\idx_\Fcal(e)$, i.e.,
$e$ is added to $G_j$ to form $G_{j+1}$
and $x$ is at level $j$.
Based on the definition of nearest common ancestors,
we have that
$v_1,v_2$ are descendants of different children of $x$.
Let $T_1,T_2$ be the two trees rooted at the two children of $x$
in $\MF(\Fcal)$ respectively.
WLOG, we can assume that $v_1$ is in $T_1$ and $v_2$ is in $T_2$ (see Figure~\ref{fig:v-v-sw}).
Since $x$ is at level $j$,
we can view
$T_1,T_2$ as trees
in $\MF^j(\Fcal)$.
Let $C_1,C_2$ be the connected components in $G_j$
corresponding to $T_1,T_2$ respectively
(see Definition~\ref{dfn:MF}).
We have that $v_1\in C_1$ and $v_2\in C_2$.
We then observe that as the simplices are added in a graph filtration,
each connected component contains only one unpaired vertex 
which is the  oldest one~\cite{edelsbrunner2000topological}.
Since $v_1$, $v_2$ are both unpaired when $e$ is added to $\Fcal$,
we must have that $v_1$ is the oldest vertex of $C_1$
and $v_2$ is the oldest vertex of a $C_2$.
Then, when $e$ is added in $\Fcal$,  $v_2$ must be paired with $e$
because $v_2$ is   younger than $v_1$
(see the pairing in the persistence algorithm~\cite{edelsbrunner2000topological}).
However, after the switch, 
$v_1$ must be paired with $e$ in $\Fcal'$
because $v_1$ is now  younger than $v_2$.
Therefore, the pairing changes after the switch 
and we have finished the proof of the `if' part of the proposition.

We now prove the `only if' part of the proposition.
Suppose that the pairing changes after the switch.
First, if $v_1$, $v_2$ are in different trees in $\MF(\Fcal)$,
then the two vertices are in different connected components in  $G:=G_m$.
The pairings for $v_1$ and $v_2$ are completely independent in the filtrations and 
therefore cannot change due to the switch.
Then,
let $T$ be the subtree of $\MF(\Fcal)$ rooted at $x$
and let $j=\idx_\Fcal(e)$.
Similarly as before,
we have that $v_1$ is in a connected component $C_1$ of $G_j$
and $v_2$ is in a connected component $C_2$ of $G_j$
for $C_1\neq C_2$.
For contradiction, suppose instead that 
at least one of
$v_1$, $v_2$ is paired when $e$ is added to $G_j$ in $\Fcal$.
If $v_1$ is paired  when $e$ is added in $\Fcal$,
then let $u$ be the oldest vertex of $C_1$.
Notice that $u\neq v_1$ because the oldest vertex of $C_1$ must be unpaired
when $e$ is added in $\Fcal$.
We notice that the pairing for vertices in $C_1\setminus\Set{u}$
only depends on the index order of these vertices in the filtrations,
before and after the switch.
Since the indices for simplices in a filtration are unique
and $v_2\not\in C_1$, 
changing the index of $v_1$ from $i-1$ to $i$
does not change 
the index order of vertices in $C_1\setminus\Set{u}$.
Therefore, the pairing for vertices in $C_1\setminus\Set{u}$
stays the same after the switch, 
which means that the pairing of $v_1$ does not change.
This contradicts the assumption that the pairing for $v_1$,
$v_2$ changes due to the switch.
If $v_2$ is paired  when $e$ is added in $\Fcal$,
we can reach a similar contradiction,
and the proof is done.
\end{proof}

\begin{proposition}\label{prop:std-vv-mf}
For the switch operation in Equation~(\ref{eqn:std-switch})
where $v_1:=\sG$ and $v_2:=\tG$ are both vertices,
$\MF(\Fcal)$ and $\MF(\Fcal')$ have the same structure,
with the only difference being on the levels of $v_1$, $v_2$
due to the index change of simplices.
\end{proposition}
\begin{proof}
The structure of the merge forest for a graph filtration
only depends on how the edges merge different connected components for the filtration.
Thus, switching two vertices does not alter the structure of the merge forest.
\end{proof}

\subsubsection{Justification for vertex-edge switch} 
\label{sec:pf-v-e-switch}
We have the following Proposition~\ref{prop:std-ve}:
\begin{proposition}\label{prop:std-ve}
For the switch operation in Equation~(\ref{eqn:std-switch})
where $v:=\sG$ is a vertex and $e:=\tG$ is an edge,
$\MF(\Fcal)$ and $\MF(\Fcal')$ have the same structure,
with the only difference being on the levels of $v$ and (possibly) $e$
due to the index change of simplices.
Moreover,  the pairings for $\Fcal$ and $\Fcal'$ stay the same.
\end{proposition}
\begin{remark*}
Notice that $e$ may not correspond to a node in the merge forests in the above setting.
\end{remark*}
\begin{proof}
The fact that the pairing stays the same follows from~\cite[Section 3]{cohen2006vines},
i.e., the `$R$ matrix' is not reduced iff the two transposed simplices have the same dimension.

For the structure of the merge forests,
we argue on the following cases:
\begin{description}
\item[$e$ is negative:]
Consider $\MF^{i+1}(\Fcal)$ as in Figure~\ref{fig:v-e-sw}.
Since $v$ is not a vertex of $e$,
$v$ must be an isolated node in $\MF^{i+1}(\Fcal)$.
From the figure, it is evident that the structure of 
$\MF^{i+1}(\Fcal)$ and $\MF^{i+1}(\Fcal')$ is the same
where the only change is the levels of $v$ and $e$.
Since the remaining construction of  $\MF(\Fcal)$ and $\MF(\Fcal')$
from $\MF^{i+1}(\Fcal)$ and $\MF^{i+1}(\Fcal')$ (respectively)
follows the same process, we have our conclusion.
\item[$e$ is positive:]
Since adding $e$ does not alter the connected components in $\Fcal$ and $\Fcal'$,
$e$ does not correspond to a node in the merge forests and the proposition is obvious.
\qedhere
\end{description}
\end{proof}

\begin{figure}[!tbh]
  \centering
  \includegraphics[width=0.5\linewidth]{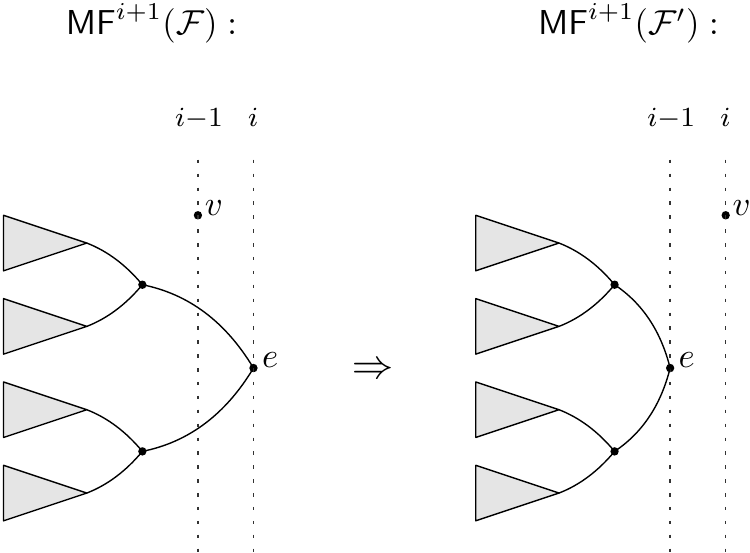}
  \caption{Parts of the sub-forests $\MF^{i+1}(\Fcal)$, $\MF^{i+1}(\Fcal')$. Node level increases from left to right.}
  \label{fig:v-e-sw}
\end{figure}

\subsubsection{Justification for edge-vertex switch}
\label{sec:pf-e-v-switch}
The behavior of an edge-vertex switch is symmetric to a vertex-edge switch
and the justification is similar as done in Section~\ref{sec:pf-v-e-switch}.

\subsubsection{Justification for edge-edge switch}
\label{sec:pf-e-e-switch}
For an edge-edge switch, rewrite the switch in Equation~(\ref{eqn:std-switch})
as follows:
\begin{equation}\label{eqn:std-switch-ee}
\begin{tikzpicture}[baseline=(current  bounding  box.center)]
\tikzstyle{every node}=[minimum width=24em]
\node (a) at (0,0) {$\Fcal: G_0 \incto
\cdots
\incto
G_{i-1}\inctosp{e_1} 
G_i 
\inctosp{e_2} G_{i+1}
\incto
\cdots \incto G_\filtcnt$}; 
\node (b) at (0,-0.6){$\Fcal': G_0 \incto
\cdots
\incto 
G_{i-1}\inctosp{e_2} 
G'_i 
\inctosp{e_1} G_{i+1}
\incto
\cdots \incto G_\filtcnt$};
\path[->] (a.0) edge [bend left=90,looseness=1.5,arrows={-latex},dashed] (b.0);
\end{tikzpicture}
\end{equation}

\begin{proposition}
For the switch operation in Equation~(\ref{eqn:std-switch-ee})
where $e_1,e_2$ are positive  in $\Fcal$,
$\MF(\Fcal')=\MF(\Fcal)$ and the pairings for the two filtrations stay the same.
\end{proposition}
\begin{proof}
The edges
$e_1,e_2$ stay positive after the switch and so the pairing stays 
the same as positive edges are always unpaired.
The merge forest stays the same because the positive edges cause no change to the connectivity
in a filtration.
\end{proof}

\begin{proposition}
\label{prop:e-e-switch-pos-neg}
For the switch operation in Equation~(\ref{eqn:std-switch-ee})
where $e_1$ is positive and $e_2$ is negative  in $\Fcal$,
after the switch,
$e_1$ stays positive and $e_2$ stays negative  in $\Fcal'$.
Moreover,
$\MF(\Fcal)$ and $\MF(\Fcal')$ have the same structure,
with the only difference being on the level of $e_2$
due to the index change of simplices.
The pairings for $\Fcal$ and $\Fcal'$ also stay the same.
\end{proposition}
\begin{proof}
First consider adding $e_2$ in $\Fcal$,
and suppose that $e_2=(u,v)$. 
Since $e_2$ is negative in $\Fcal$,
the vertices $u,v$ are in different connected components $C_1,C_2$
of $G_i$ (see~\cite{edelsbrunner2000topological}).
Moreover, since $e_1$ is positive in $\Fcal$, 
the connectivity of $G_{i-1}$ and $G_i$ is the same.
Then, when we add $e_2$ to $G_{i-1}$ in $\Fcal'$,
we can also consider $C_1,C_2$ as connected components of $G_{i-1}$
and consider $u,v$ to be vertices in $C_1,C_2$ (respectively).
Therefore, adding $e_2$ to $G_{i-1}$ in $\Fcal'$
connects the two components $C_1,C_2$,
which indicates that $e_2$ is still negative in $\Fcal'$.
Notice that $e_1$ in $\Fcal'$ also creates a 1-cycle when added because $G_i\subseteq G_{i+1}$
(indicating that $e_1$ stays positive in $\Fcal'$),
whose addition does not change the connectivity.
Therefore, the variation of connected components in $\Fcal$ and $\Fcal'$ is the same,
and we have that $\MF(\Fcal),\MF(\Fcal')$ have the same structure.
The fact that the pairings of $\Fcal$ and $\Fcal'$ stay the same
follows from Case 4 of the algorithm presented in~\cite[Section 3]{cohen2006vines}.
\end{proof}

\begin{proposition}\label{prop:std-switch-ee-neg-pos}
For the switch operation in Equation~(\ref{eqn:std-switch-ee})
where $e_1$ is negative and $e_2$ is positive  in $\Fcal$,
there are two different situations:
\begin{description}
\item[$e_1$ is in a 1-cycle in $G_{i+1}$:]
In this case,
$\MF(\Fcal)$ and $\MF(\Fcal')$ have the same structure,
with the only difference that the node corresponding to $e_1$ in $\MF(\Fcal)$
now corresponds to $e_2$ in $\MF(\Fcal')$.
Furthermore,
the vertex paired with $e_1$ in $\Fcal$ is now paired with $e_2$ in $\Fcal'$,
and $e_1$ becomes positive (unpaired) in $\Fcal'$.
\item[$e_1$ is not in a 1-cycle in $G_{i+1}$:]
In this case,
$\MF(\Fcal)$ and $\MF(\Fcal')$ have the same structure,
with the only difference being on the levels of $e_1$
due to the index change of simplices.
The pairings for $\Fcal$ and $\Fcal'$ stay the same.
\end{description}
\end{proposition}
\begin{proof}
First suppose that 
$e_1$ is in a 1-cycle in $G_{i+1}$.
Then adding $e_1$ to $G'_{i}$ in $\Fcal'$ does not
change the connectivity of $G'_{i}$
because $e_1$ is positive in $\Fcal'$.
But we know that two connected components $C_1,C_2$ merge into a single one
from $G_{i-1}$ to $G_{i+1}$, following the assumptions on $\Fcal$  (see Figure~\ref{fig:e1-e2-same-conns}). 
So we must have that 
adding $e_2$ to $G_{i-1}$ in $\Fcal'$ causes the merge of $C_1,C_2$.
Hence, the change on the merges forests and pairings 
as described in the proposition is true.
See Figure~\ref{fig:e1-e2-same-conns} for an illustration of the situation described above.

Now suppose that 
$e_1$ is not in a 1-cycle in $G_{i+1}$.
Then, it is obvious that a 1-cycle $z\subseteq G_{i+1}$ created by the addition of $e_2$ in $\Fcal$
does not contain $e_1$.
Therefore, we have $z\subseteq G'_{i}$
because the only difference of $G'_i$ with 
$G_{i+1}$ is the missing of $e_1$.
Then we have that $e_2$ is positive and $e_1$ is negative 
in $\Fcal'$.
Since positive edges do not alter the connectivity of graphs,
the second part of the proposition follows.
\end{proof}

Proposition~\ref{prop:std-switch-ee-neg-neg-ez}
and~\ref{prop:std-switch-ee-neg-neg-hard}
justify the case where $e_1,e_2$ are both negative:

\begin{proposition}
\label{prop:std-switch-ee-neg-neg-ez}
For the switch operation in Equation~(\ref{eqn:std-switch-ee})
where $e_1,e_2$ are negative,
if the corresponding node of $e_1$ is not a child of the corresponding node of $e_2$
in $\MF(\Fcal)$, 
then
$\MF(\Fcal)$ and $\MF(\Fcal')$ have the same structure,
with the only difference being on the levels of $e_1,e_2$
due to the index change of simplices.
The pairings for $\Fcal$ and $\Fcal'$ also stay the same.
\end{proposition}
\begin{proof}
Following the assumptions in the proposition, 
we have that the connected components $C_1,C_2$ 
that $e_1$ merges 
and the connected components $C_3,C_4$ 
that $e_2$ merges 
are all different
and can all be considered as connected components in $G_{i-1}$.
The proposition is then evident from this fact.
\end{proof}

\begin{proposition}\label{prop:std-switch-ee-neg-neg-hard}
For the switching of two edges $e_1,e_2$
where $e_1,e_2$ are both negative
and the corresponding node of $e_1$ is a child of the corresponding node of $e_2$
in $\MF(\Fcal)$, 
Algorithm~\ref{alg:std-switch-abs} makes the correct changes on the merge forest and the  simplex pairing.
\end{proposition}
\begin{proof}
From Figure~\ref{fig:e-e-sw} and~\ref{fig:C1-C2-C3},
it is not hard to see that Algorithm~\ref{alg:std-switch-abs} makes the correct changes on the merge forest 
for the situation in the proposition.
Therefore, we only need to show that 
Algorithm~\ref{alg:std-switch-abs} makes the correct changes on the  pairing.
Since indices of 
$u,v,w$ as defined in  Algorithm~\ref{alg:std-switch-abs}
stay the same in $\Fcal$ and $\Fcal'$,
for these vertices,
we use, e.g., $\idx(u)$ to denote the index of $u$ in the filtrations.
We have the following cases (notice that $u$ is always paired with $e_1$ in $\Fcal$):
\begin{description}
   \item[$e_2$ directly connects $C_2,C_3$ as in Figure~\ref{fig:C1-C2-C3-1}:]
   Suppose that $\idx(w)<\idx(v)$.
   After the addition of $e_2$ in $\Fcal'$,
   $w$ is the representative (oldest vertex) for the
   merged component $C$ of $C_2$ and $C_3$ in $G'_i$.
   Subsequently, $u$ is paired with $e_1$ in $\Fcal'$ 
   due to the merge of $C$ and $C_1$ because $\idx(w)<\idx(v)<\idx(u)$.
   It is then evident that the pairings for $\Fcal$ and $\Fcal'$ do not change
   with the current assumptions.
   If $\idx(v)<\idx(w)$, by similar arguments, we also have that the pairing does not change.
   \item[$e_2$ directly connects $C_1,C_3$ as in Figure~\ref{fig:C1-C2-C3-2}:]
    If 
    $\idx(w)<\idx(u)$,
    when adding $e_2$ in $\Fcal'$, $u$ is paired with $e_2$
    due to the merge of $C_1$ and $C_3$.
    Therefore,
    the pairings for $\Fcal$ and $\Fcal'$ change
    with the current assumptions.
    If $\idx(w)>\idx(u)$, 
    by similar arguments, we have that the pairing does not change.
    \qedhere
\end{description}
\end{proof}

From the justifications above,
we conclude the following:

\begin{theorem}
Algorithm~\ref{alg:std-switch-abs} correctly updates the pairing
and the merge forest for a switch operation in $O(\log m)$ time.
\end{theorem}

\section{Computing graph zigzag persistence}\label{sec:gzz-non-up}
In this section, we show how the usage of  
the Link-Cut Tree~\cite{sleator1981data},
can improve the computation of graph zigzag persistence. For this purpose,
we combine two recent results: 

\begin{itemize}
    \item 
We 
show in~\cite{DBLP:conf/esa/DeyH22} that a given zigzag filtration can be converted into
 a standard filtration for a fast computation of zigzag barcode.
 \item
Yan et al.~\cite{yan2021link} show that the extended persistence
of a given graph filtration can be computed by using operations only on trees. 
\end{itemize}

Building on the work of~\cite{DBLP:conf/esa/DeyH22}, we first convert a
given simplex-wise graph zigzag filtration into a \textit{cell}-wise \textit{up-down} 
filtration, where
all insertions occur before deletions. Then, using the extended
persistence algorithm of Yan et al.~\cite{yan2021link} on the up-down filtration with the  Link-Cut Tree~\cite{sleator1981data} data structure,
we obtain an improved $O(\filtcnt\log \filtcnt)$ algorithm for computing graph zigzag persistence.

\subsection{Converting to up-down filtration}
\label{sec:conv-2-ud}
First, we recall the necessary set up from~\cite{DBLP:conf/esa/DeyH22} relevant to our purpose.
The algorithm, called \textsc{FastZigzag}~\cite{DBLP:conf/esa/DeyH22}, builds filtrations 
on extensions of simplicial complexes called 
\textit{$\DG$-complexes}~\cite{hatcher2002algebraic},
whose building blocks are
called \textit{cells} or \textit{$\DG$-cells}.
Notice that  1-dimensional $\DG$-complexes are nothing but graphs with \emph{parallel edges}  (also termed as \emph{multi-edges})~\cite{DBLP:conf/esa/DeyH22}.

Assume a \textit{simplex}-wise graph zigzag filtration 
\[\Fcal:
\emptyset=
G_0\leftrightarrowsp{\fsimp{}{0}} G_1\leftrightarrowsp{\fsimp{}{1}}
\cdots 
\leftrightarrowsp{\fsimp{}{\filtcnt-1}} G_\filtcnt
=\emptyset
\]
consisting of simple graphs as input.
We convert $\Fcal$ into the following \textit{cell}-wise up-down~\cite{carlsson2009zigzag-realvalue} filtration
consisting of graphs with parallel edges:
\begin{equation}\label{eqn:ud}
\Ud:\emptyset=\hat{G}_0\inctosp{\hatfsimp_0}
\hat{G}_1\inctosp{\hatfsimp_{1}}
\cdots\inctosp{\hatfsimp_{\simpcnt-1}} 
\hat{G}_{\simpcnt}
\bakinctosp{\hatfsimp_{\simpcnt}}
\hat{G}_{\simpcnt+1}
\bakinctosp{\hatfsimp_{\simpcnt+1}}
\cdots
\bakinctosp{\hatfsimp_{\filtcnt-1}} \hat{G}_{\filtcnt}=\emptyset.
\end{equation}
Cells $\hatfsimp_0,\hatfsimp_1,\ldots,\hatfsimp_{\simpcnt-1}$ 
are \emph{uniquely identified} copies of vertices and edges added in $\Fcal$ 
with the addition order preserved
(notice that the same vertex or edge could be repeatedly
added and deleted in $\Fcal$).
Cells 
$\hatfsimp_{\simpcnt},\hatfsimp_{\simpcnt+1},\ldots,\hatfsimp_{\filtcnt-1}$ 
are \emph{uniquely identified} copies of vertices and edges deleted in $\Fcal$,
with the order also preserved.
Notice that $\filtcnt=2\simpcnt$
because an added simplex
must be eventually deleted in ${\Fcal}$.
\opt{showWADSdel}{Figure~\ref{fig:convert} illustrates an example for converting an 
input graph zigzag filtration $\Fcal$
into an up-down filtration $\Ud$ with parallel edges.
In Figure~\ref{fig:convert},
the edge $e$ is added twice in $\Fcal$
in which the first addition corresponds to $\hat{e}_1$ in $\Ud$
and the second addition corresponds to $\hat{e}_2$ in $\Ud$.}

\opt{showWADSdel}{
\begin{figure}[!tb]
  \centering
  \includegraphics[width=\linewidth]{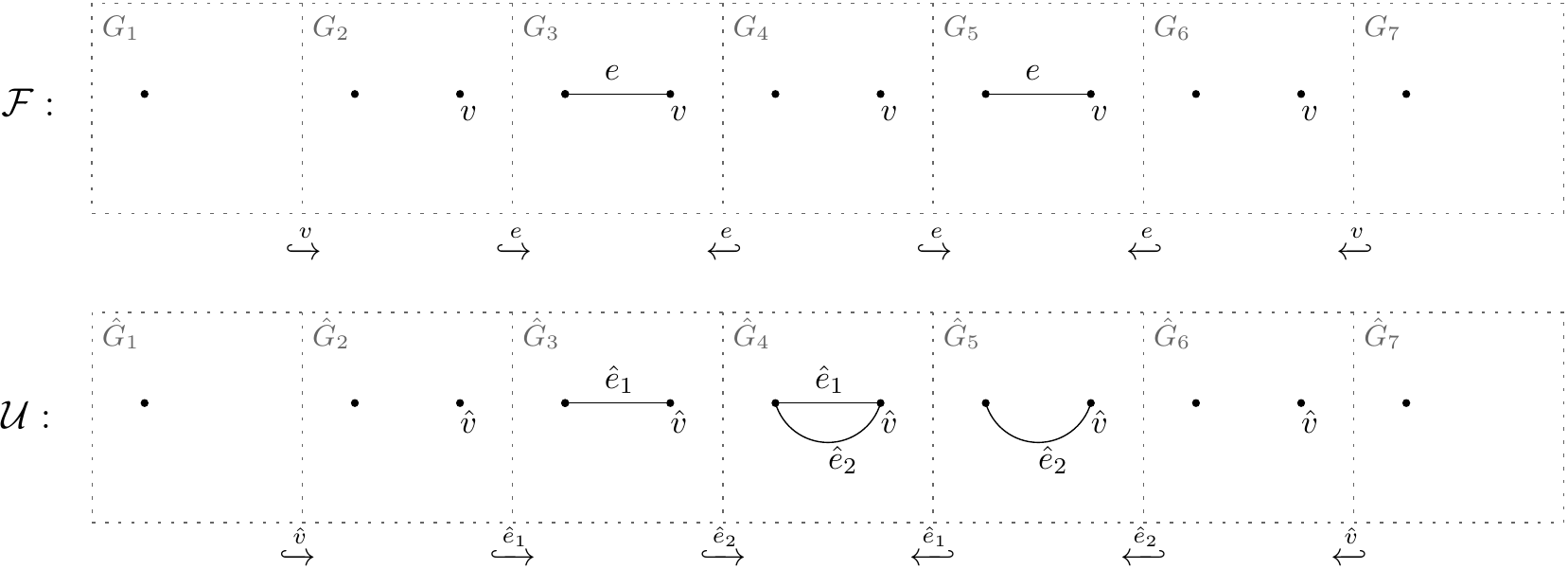}
  \caption{An example of converting a graph zigzag filtration $\Fcal$ to an up-down filtration $\Ud$.}
  \label{fig:convert}
\end{figure}}

\begin{definition}\label{dfn:creator-destroyer}
In $\Fcal$ or $\Ud$, 
let each addition or deletion 
be uniquely identified by its index in the filtration,
e.g., index of $G_i\leftrightarrowsp{\fsimp{}{i}}G_{i+1}$ in $\Fcal$ is $i$.
Then, 
the \defemph{creator} of an interval $[b,d]\in\Pers_*(\Fcal)\text{ or }\Pers_*(\Ud)$
is an addition/deletion indexed at $b-1$,
and the \defemph{destroyer} of $[b,d]$
is an addition/deletion indexed at $d$.
\end{definition}

As stated previously, 
each $\hatfsimp_i$ in $\Ud$ for $0\leq i<\simpcnt$ corresponds
to an addition in $\Fcal$, and
each $\hatfsimp_i$ for $\simpcnt\leq i<\filtcnt$ corresponds
to a deletion in $\Fcal$.
This naturally defines a bijection $\phi$
from the additions and deletions in $\Ud$ to the additions and deletions in $\Fcal$.
Moreover, for simplicity, we let the domain and codomain of $\phi$
be the sets of indices for the additions and deletions.
The interval mapping in~\cite{DBLP:conf/esa/DeyH22} (which uses the \emph{Mayer-Vietoris Diamond}~\cite{carlsson2010zigzag,carlsson2009zigzag-realvalue}) can be summarized
as follows:

\begin{theorem}
\label{thm:fzz-intv-map}
Given $\Pers_*(\Ud)$, one can retrieve $\Pers_*(\Fcal)$ 
using the following bijective mapping
from $\Pers_*(\Ud)$
to $\Pers_*(\Fcal)${\rm:}
an 
interval $[b,d]\in\Pers_p(\Ud)$ 
with a creator indexed at $b-1$ and a destroyer indexed at $d$
is mapped to an interval $I\in \Pers_*(\Fcal)$ with
the same creator and destroyer indexed at 
$\phi(b-1)$ and $\phi(d)$ respectively.
Specifically, 
\begin{itemize}
    \item 
If $\phi(b-1)<\phi(d)$, then $I=[\phi(b-1)+1,\phi(d)]\in \Pers_p(\Fcal)$,
where $\phi(b-1)$ indexes the creator and $\phi(d)$ indexes the destroyer.
\item
Otherwise,
$I=[\phi(d)+1,\phi(b-1)]\in \Pers_{p-1}(\Fcal)$,
where $\phi(d)$ indexes the creator and $\phi(b-1)$ indexes the destroyer.
\end{itemize}
\end{theorem}

Notice the decrease in the dimension of the mapped interval in $\Pers_{*}(\Fcal)$
when $\phi(d)<\phi(b-1)$
(indicating a swap on the roles of the creator and destroyer).

While Theorem~\ref{thm:fzz-intv-map} suggests a simple mapping rule for
$\Pers_*(\Ud)$ and $\Pers_*(\Fcal)$,
we further interpret the mapping in terms of the different types of intervals
in zigzag persistence. We define the following:

\begin{definition}
\label{dfn:open-close-bd}
Let
$\Lcal: \emptyset=K_0 \leftrightarrow K_1 \leftrightarrow 
\cdots \leftrightarrow K_\ell=\emptyset$ be a zigzag filtration.
For any $[\birth,\death]\in \Pers_*(\Lcal)$,
the birth index $\birth$ is 
\defemph{closed} 
if $K_{\birth-1}\incto K_\birth$ is a 
forward 
inclusion;
otherwise, $\birth$ is  \defemph{open}.
Symmetrically, 
the death index $\death$ is 
\defemph{closed}
if $K_{\death}\bakincto K_{\death+1}$ is a backward inclusion;
otherwise, $\death$ is  \defemph{open}.
The types of the birth/death ends
classify intervals in $\Pers_*(\Lcal)$ into four types: 
\defemph{closed-closed}, \defemph{closed-open}, \defemph{open-closed}, and \defemph{open-open}. 
\end{definition}

Table~\ref{tab:fzz-map-types} breaks down the bijection between $\Pers_*(\Ud)$ and $\Pers_*(\Fcal)$
into mappings for the different types,
where $\Pers_0^{\text{co}}(\Ud)$ denotes the set of closed-open intervals in $\Pers_0(\Ud)$
(meanings of other symbols can be derived similarly).
\begin{table}[htb!]
    \centering
    \caption{Mapping of different types of intervals for $\Pers_*(\Ud)$ and $\Pers_*(\Fcal)$}
    \begin{tabular}{ccc}
    \midrule
        $\Ucal$ & & $\Fcal$ \\
        \midrule
        $\Pers_0^{\text{co}}(\Ud)$ & $\leftrightarrow$ & $\Pers_0^{\text{co}}(\Fcal)$\\
        $\Pers_0^{\text{oc}}(\Ud)$ & $\leftrightarrow$ & $\Pers_0^{\text{oc}}(\Fcal)$\\
        $\Pers_0^{\text{cc}}(\Ud)$ & $\leftrightarrow$ & $\Pers_0^{\text{cc}}(\Fcal)$\\
        $\Pers_1^{\text{cc}}(\Ud)$ & $\leftrightarrow$ & $\Pers_0^{\text{oo}}(\Fcal)\union\Pers_1^{\text{cc}}(\Fcal)$\\
        \midrule
    \end{tabular}
    \label{tab:fzz-map-types}
\end{table}

We notice the following:
\begin{itemize}
    \item $\Pers_*(\Ud)$ has no open-open intervals because there are no additions after
    deletions in $\Ud$.
    \item $\Pers_1(\Ud)$ and $\Pers_1(\Fcal)$ contain only closed-closed intervals
    because graph filtrations have no triangles.
    \item $[b,d]\in\Pers_1^{\text{cc}}(\Ud)$ is mapped to 
    an interval in $\Pers_0^{\text{oo}}(\Fcal)$ when $\phi(d)<\phi(b-1)$.
\end{itemize}

\opt{showWADSdel}{
\begin{example*}
The interval mapping for the example in Figure~\ref{fig:convert} is as follows:
\begin{align*}
&[2,2]\in \Pers^\text{co}_0(\Fcal)\leftrightarrow[2,2]\in \Pers^\text{co}_0(\Ud),\quad
[6,6]\in \Pers^\text{oc}_0(\Fcal)\leftrightarrow[6,6]\in \Pers^\text{oc}_0(\Ud),\\
&[4,4]\in \Pers^\text{oo}_0(\Fcal)\leftrightarrow[4,4]\in \Pers^\text{cc}_1(\Ud),\quad
[1,7]\in \Pers^\text{cc}_0(\Fcal)\leftrightarrow[1,7]\in \Pers^\text{cc}_1(\Ud).
\end{align*}
\end{example*}}

\subsection{Extended persistence algorithm for graphs}
\label{sec:tree-extd-pers}
 Yan et al.~\cite{yan2021link} present an extended persistence algorithm
for graphs in a neural network setting (see also~\cite{ZMD22}) which runs in quadratic time. We adapt it
to computing up-down zigzag persistence while improving its time complexity
with a Link-Cut tree~\cite{sleator1981data} data structure.

Definition~\ref{dfn:creator-destroyer}
indicates that for 
the up-down filtration in Equation~(\ref{eqn:ud}), 
$\Pers_*(\Ud)$ can be considered
as generated from the \emph{cell pairs} similar to the simplex pairs in standard
persistence~\cite{edelsbrunner2000topological}.
Specifically, 
an interval $[b,d]\in\Pers_*(\Ud)$
is generated from the pair $(\hat{\sG}_{b-1},\hat{\sG}_{d})$.
While each cell appears twice in $\Ud$ (once  added and once deleted),
we notice that it should be clear from the context whether 
a cell
in a pair
refers to its addition or deletion.
We then have the following:
\begin{remark}\label{rmk:ud-asc-desc}
Every vertex-edge
pair for a closed-open interval in $\Pers_0(\Ud)$ comes from the
\emph{ascending} part $\Ud_u$ of the filtration $\Ud$,
and every edge-vertex pair for an open-closed interval in $\Pers_0(\Ud)$ 
comes from
the \emph{descending} part $\Ud_d$. 
These ascending and descending parts are as shown below:
\begin{gather}\label{eqn:ud-asc-desc}
\begin{aligned}
&\Ud_u:\emptyset=\hat{G}_0\inctosp{\hatfsimp_0} 
\hat{G}_1\inctosp{\hatfsimp_{1}}
\cdots\inctosp{\hatfsimp_{\simpcnt-1}} 
\hat{G}_{\simpcnt},\\
&\Ud_d:\emptyset= \hat{G}_{\filtcnt}\inctosp{\hatfsimp_{\filtcnt-1}}
\hat{G}_{\filtcnt-1}
\inctosp{\hatfsimp_{\filtcnt-2}}
\cdots
\inctosp{\hatfsimp_{\simpcnt}} \hat{G}_{\simpcnt}.
\end{aligned}
\end{gather}
\end{remark}

We  first run the standard persistence algorithm with the Union-Find
data structure on $\Ud_u$ and $\Ud_d$ to obtain all pairs between vertices 
and edges in $O(\simpcnt\,\alpha(\simpcnt))$ time, retrieving closed-open and open-closed
intervals in $\Pers_0(\Ud)$. 
We also have the following:
\begin{remark}\label{rmk:cc-intv}
Each closed-closed interval in $\Pers_0(\Ud)$
is given by pairing the first vertex in $\Ud_u$,
that comes from a connected component $C$ of $\hat{G}_{\simpcnt}$, and the
first vertex in $\Ud_d$ coming from $C$. 
There is no extra computation necessary
for this type of pairing. 
\end{remark}

\begin{remark}\label{rmk:ee-pair}
Each closed-closed interval in $\Pers_1(\Ud)$
is given by
an edge-edge pair in $\Ud$,
in which one edge 
is a positive edge from
the ascending filtration $\Ud_u$ and the other 
is a positive edge from the descending filtration
$\Ud_d$. 
\end{remark}

To compute the edge-edge pairs,
the algorithm scans $\Ud_d$
and keeps track of whether an edge is positive or negative.
For every positive edge $e$ in $\Ud_d$, it finds the cycle $c$
that is created the earliest in $\Ud_u$ containing $e$ and then
pairs $e$ with the youngest edge $e'$ of $c$ added in $\Ud_u$, which creates $c$ in $\Ud_u$.
To determine $c$ and $e'$,  we use  the following
procedure from~\cite{yan2021link}:

\begin{algr}\label{alg:ext-pers}
\begin{enumerate}
   \item[]
    \item 
    Maintain a spanning forest $T$ of
$\hat{G}_\simpcnt$ while processing $\Ud_d$.
Initially, $T$ consists
of all vertices of $\hat{G}_\simpcnt$ and all negative edges in $\Ud_d$.
    \item For every positive edge $e$ in $\Ud_d$  (in the order of the filtration):
    \begin{enumerate}
    \item 
Add $e$ to $T$ and
check the \textit{unique} cycle $c$ formed by $e$ in $T$. 

\item 
Determine the edge $e'$ which is the youngest
edge of $c$ with respect to the filtration $\Ud_u$. The
edge $e'$ has to be positive in $\Ud_u$.

\item 
Delete $e'$ from $T$.
This
maintains $T$ to be a tree all along. 

\item 
Pair the positive edge $e$ from
$\Ud_d$ with the positive edge $e'$ from $\Ud_u$.
   \end{enumerate}
\end{enumerate}
\end{algr}

We propose to implement the above algorithm by maintaining $T$ as a Link-Cut
Tree~\cite{sleator1981data}, which is a dynamic data structure 
allowing the following
operations in $O(\log N)$ time ($N$ is the number of nodes in the trees): 
(i) insert or delete a node or an edge
from the Link-Cut Trees; (ii)
find the maximum-weight edge on a path in the trees. 
Notice that
for the edges in $T$, we let their weights
equal to their indices in $\Ud_u$. 

We build $T$ by first 
inserting 
all vertices of $\hat{G}_\simpcnt$ and all negative edges in $\Ud_d$
into $T$ 
in $O(m\log m)$ time.
Then,
for every positive
edge $e$ in $\Ud_d$,  
find the maximum-weight edge $\epsilon$ in the unique path in $T$
connecting
the two endpoints of $e$ in $O(\log m)$  time.
Let $e'$ be the edge in $\Set{e,\epsilon}$ whose index in $\Ud_u$ is greater 
(i.e., $e'$ is the younger one in $\Ud_u$).
Pair $e'$ with $e$ to form a closed-closed interval in $\Pers_1(\Ud)$. After this, 
delete $e'$ from $T$ and insert
 $e$ into $T$, which takes $O(\log m)$ time.  
Therefore, processing the entire
filtration $\Ud_d$ and getting all closed-closed intervals in $\Pers_1(\Ud)$ 
takes $O(m\log m)$ time in total.

\begin{theorem}
For a {simplex}-wise graph zigzag filtration $\Fcal$ with $m$ additions and deletions,
$\Pers_*(\Fcal)$ can be computed in $O(m\log m)$ time. 
\end{theorem}
\begin{proof}
We first convert $\Fcal$ into the up-down filtration $\Ud$
in $O(m)$ time~\cite{DBLP:conf/esa/DeyH22}.
We then compute $\Pers_*(\Ud)$ in $O(m\log m)$ time using the algorithm described in 
this section.
Finally, we convert $\Pers_*(\Ud)$ to $\Pers_*(\Fcal)$ using the process in Theorem~\ref{thm:fzz-intv-map},
which takes $O(m)$ time.
Therefore, computing $\Pers_*(\Fcal)$ takes  $O(m\log m)$ time. 
\end{proof}

\section{Updating graph zigzag persistence}
\label{sec:zz-switch}

In this section, we describe the update of persistence for switches on 
graph zigzag filtrations. In~\cite{dey2021updating}, we considered the
updates in zigzag filtration for general simplicial complexes. Here, we focus
on the special case of graphs, for which we find more efficient algorithms
for switches.
In a similar vein to the switch operation on standard filtrations~\cite{cohen2006vines} (see also Section~\ref{sec:std-switch}),
a switch on a zigzag filtration swaps two consecutive simplex-wise inclusions.
Based on the directions of the inclusions, we have the following four types of switches (as defined in~\cite{dey2021updating}),
where $\Fcal,\Fcal'$ are both simplex-wise graph zigzag filtrations  starting and ending with empty graphs:

\begin{itemize}
    \item 
\textit{{Forward}} switch is the counterpart of the switch on standard filtrations, 
which swaps two forward inclusions (i.e., additions)
and also requires $\sG\nsubseteq\tG$:
\begin{equation}\label{eqn:fwd-switch}
\begin{tikzpicture}[baseline=(current  bounding  box.center)]
\tikzstyle{every node}=[minimum width=24em]
\node (a) at (0,0) {$\Fcal:G_0 \leftrightarrow\cdots\leftrightarrow G_{i-1}\inctosp{\sG}G_i\inctosp{\tG}G_{i+1}\leftrightarrow\cdots\leftrightarrow G_\filtcnt$}; 
\node (b) at (0,-0.6){$\Fcal':G_0\leftrightarrow\cdots\leftrightarrow G_{i-1}\inctosp{\tG} G'_i\inctosp{\sG} G_{i+1}\leftrightarrow\cdots\leftrightarrow G_\filtcnt$};
\path[->] (a.0) edge [bend left=90,looseness=1.5,arrows={-latex},dashed] (b.0);
\end{tikzpicture}
\end{equation}

\item
\textit{{Backward}} switch
is the symmetric version of forward switch,
requiring $\tG\not\subseteq\sG$:
\begin{equation}
\label{eqn:bak-switch}
\begin{tikzpicture}[baseline=(current  bounding  box.center)]
\tikzstyle{every node}=[minimum width=24em]
\node (a) at (0,0) {$\Fcal: G_0 \leftrightarrow
\cdots
\leftrightarrow 
G_{i-1}\bakinctosp{\sG} 
G_i 
\bakinctosp{\tG} G_{i+1}
\leftrightarrow
\cdots \leftrightarrow G_\filtcnt$}; 
\node (b) at (0,-0.6){$\Fcal': G_0 \leftrightarrow
\cdots
\leftrightarrow 
G_{i-1}\bakinctosp{\tG} 
G'_i 
\bakinctosp{\sG} G_{i+1}
\leftrightarrow
\cdots \leftrightarrow G_\filtcnt$};
\path[->] (a.0) edge [bend left=90,looseness=1.5,arrows={-latex},dashed] (b.0);
\end{tikzpicture}
\end{equation}

\item
The remaining switches
swap two inclusions of opposite directions:
\begin{equation}\label{eqn:out-in-switch}
\begin{tikzpicture}[baseline=(current  bounding  box.center)]
\tikzstyle{every node}=[minimum width=24em]
\node (a) at (0,0) {$\Fcal: G_0 \leftrightarrow
\cdots
\leftrightarrow 
G_{i-1}\inctosp{\sG} 
G_i 
\bakinctosp{\tG} G_{i+1}
\leftrightarrow
\cdots \leftrightarrow G_\filtcnt$}; 
\node (b) at (0,-0.6){$\Fcal': G_0 \leftrightarrow
\cdots
\leftrightarrow 
G_{i-1}\bakinctosp{\tG} 
G'_i 
\inctosp{\sG} G_{i+1}
\leftrightarrow
\cdots \leftrightarrow G_\filtcnt$};
\path[-] (a.0) edge [bend left=90,looseness=1.5,arrows={latex-latex},dashed] (b.0);
\end{tikzpicture}
\end{equation}
The switch from $\Fcal$ to $\Fcal'$ is called
an \textit{{outward}} 
switch
and the switch from $\Fcal'$ to $\Fcal$ is called an \textit{{inward}} switch.
We also require $\sG\neq\tG$
because if $\sG=\tG$, e.g., for outward switch, we cannot delete $\tG$ from $G_{i-1}$ in $\Fcal'$
because $\tG\not\in G_{i-1}$.
\end{itemize}

\subsection{Update algorithms}

Instead of performing the updates in Equation~(\ref{eqn:fwd-switch}\,--\,\ref{eqn:out-in-switch}) 
directly on the graph zigzag filtrations,
our algorithms work
\emph{on the corresponding up-down filtrations} for $\Fcal$ and $\Fcal'$,
with the conversion described in Section~\ref{sec:conv-2-ud}.
Specifically, we maintain a pairing of cells for the corresponding up-down filtration,
and the pairing for the original graph zigzag filtration
can be derived from the bijection $\phi$ as defined in Section~\ref{sec:conv-2-ud}.

For outward and inward switches, the corresponding up-down filtration
before and after the switch
is the same and hence the update takes $O(1)$ time.
Moreover, the backward switch is a symmetric version of the forward switch
and the algorithm is also symmetric.
Hence, in this section we focus on how to perform the forward switch.
The symmetric behavior for backward switch is mentioned only when necessary.

For the forward switch in Equation~(\ref{eqn:fwd-switch}),
let $\Ud$ and $\Ud'$ be the corresponding up-down filtrations for $\Fcal$ and $\Fcal'$ respectively.
By the conversion in Section~\ref{sec:conv-2-ud},
there is also a forward switch 
(on the ascending part)
from $\Ud$ to $\Ud'$,
where $\hat{\sG},\hat{\tG}$ are $\DG$-cells corresponding to $\sG,\tG$ respectively:
\begin{equation}\label{eqn:fwd-switch-ud}
\begin{tikzpicture}[baseline=(current  bounding  box.center)]
\tikzstyle{every node}=[minimum width=24em]
\node (a) at (0,0) {$\Ud:\hat{G}_0\incto\cdots\incto
\hat{G}_{j-1}\inctosp{\hat{\sG}}
\hat{G}_{j}\inctosp{\hat{\tG}}
\hat{G}_{j+1}
\incto
\cdots
\incto
\hat{G}_{\simpcnt}
\bakincto
\cdots
\bakincto \hat{G}_{\filtcnt}$}; 
\node (b) at (0,-0.6){$\Ud':\hat{G}_0\incto\cdots\incto
\hat{G}_{j-1}\inctosp{\hat{\tG}}
\hat{G}'_{j}\inctosp{\hat{\sG}}
\hat{G}_{j+1}
\incto
\cdots
\incto
\hat{G}_{\simpcnt}
\bakincto
\cdots
\bakincto \hat{G}_{\filtcnt}$};
\path[->] (a.0) edge [bend left=90,looseness=1.5,arrows={-latex},dashed] (b.0);
\end{tikzpicture}
\end{equation}

We observe that the update of the different types of intervals
for up-down filtrations (see Table~\ref{tab:fzz-map-types}) can be done \emph{independently}:

\begin{itemize}
    \item 
To update the \emph{closed-open} intervals for  Equation~(\ref{eqn:fwd-switch-ud})
(which updates the closed-open intervals for  Equation~(\ref{eqn:fwd-switch})),
we run Algorithm~\ref{alg:std-switch-abs} in Section~\ref{sec:std-switch}
on the ascending part of the up-down filtration.
This is based on
descriptions in Section~\ref{sec:gzz-non-up}
(Table~\ref{tab:fzz-map-types} and
Remark~\ref{rmk:ud-asc-desc}).

\item 
A backward switch in Equation~(\ref{eqn:bak-switch})
causes a backward switch on the descending parts of the up-down filtrations,
which may change the \emph{open-closed} intervals for the filtrations.
For this,
we run Algorithm~\ref{alg:std-switch-abs} 
on the descending part of the up-down filtration.
Our update algorithm hence maintains two sets of data structures 
needed by Algorithm~\ref{alg:std-switch-abs},
for the ascending and descending parts separately.

\item
Following Remark~\ref{rmk:cc-intv} in Section~\ref{sec:gzz-non-up},
to update the \emph{closed-closed} intervals in dimension 0 for the switches,
we only need to keep track of 
the oldest vertices
in the ascending and descending parts
for each connected component of $\hat{G}:=\hat{G}_\simpcnt$ (defined in Equation~(\ref{eqn:fwd-switch-ud})).
Since indices of no more than two vertices can change in a switch,
this can be done in constant time by a simple bookkeeping.
\end{itemize}

We are now left with the update of 
the closed-closed intervals in dimension 1 
for the switch on up-down filtrations,
which are generated from the edge-edge pairs 
(see Remark~\ref{rmk:ee-pair} in Section~\ref{sec:gzz-non-up}).
To describe the update, we first present 
a high-level algorithm (Algorithm~\ref{alg:fwd-switch-ud-cc}) 
by explicitly maintaining representatives (see Definition~\ref{def:rep}) for the edge-edge pairs.
We then justify the correctness of 
Algorithm~\ref{alg:fwd-switch-ud-cc}
by showing that the algorithm correctly maintains the representatives.
After this, 
we propose an efficient implementation of 
Algorithm~\ref{alg:fwd-switch-ud-cc}
by eliminating the maintenance of representatives.
This implementation maintains MSF's for graphs
in $\sqrt{m}$ positions in the up-down filtration,
achieving the claimed time complexity of $O(\sqrt{m}\,\log m)$ per update.

As mentioned, the maintenance of the edge-edge pairs is independent 
of the maintenance of pairs generating other types of intervals,
i.e., when we perform update on one type of pairs (e.g., edge-edge pairs),
we do not need to inform the data structure maintained for updating
other types of pairs (e.g., vertex-edge pairs).
One reason is that a switch involving a vertex, which affects other types of pairs (e.g., vertex-edge pairs), 
does not affect edge-edge pairs 
(see Algorithm~\ref{alg:std-switch-abs} and~\ref{alg:std-switch-full}).
Also, it can be easily verified that for the different cases in an edge-edge switch,
the update in Algorithm~\ref{alg:std-switch-abs} and  Algorithm~\ref{alg:fwd-switch-ud-cc} (presented below)
can be conducted completely independently.
Now define the following:
\begin{definition}\label{def:rep}
For a cell-wise up-down filtration of graphs with parallel edges
\[\Lcal:\emptyset=H_0\inctosp{\varsigma_0} 
H_1\inctosp{\varsigma_{1}}
\cdots\inctosp{\varsigma_{\ell-1}} 
H_{\ell}
\bakinctosp{\varsigma_{\ell}}
H_{\ell+1}
\bakinctosp{\varsigma_{\ell+1}}
\cdots
\bakinctosp{\varsigma_{2\ell-1}} H_{2\ell}=\emptyset,\]
a \defemph{representative cycle} (or simply \defemph{representative})
for an edge-edge pair $(\varsigma_{b},\varsigma_{d})$
is a 1-cycle $z$
s.t.\ $\varsigma_{b}\in z\subseteq H_{b+1}$ and
$\varsigma_{d}\in z\subseteq H_d$.
\end{definition}

The following algorithm 
describes the general idea for
the update on edge-edge pairs:

\begin{algr}\label{alg:fwd-switch-ud-cc}
We describe the algorithm 
for the forward switch in Equation~(\ref{eqn:fwd-switch-ud}).
The procedure for a backward switch on an up-down filtration
is symmetric.
The algorithm maintains 
a set of edge-edge pairs $\Pi$ initially for $\Ud$.
It also maintains 
a representative cycle for each edge-edge pair in $\Pi$.
After the processing,
edge-edge pairs in $\Pi$ and their representatives 
correspond to $\Ud'$.
As mentioned, a switch containing a vertex makes no changes to 
the edge-edge pairs (see Algorithm~\ref{alg:std-switch-full}).
So, suppose that the switch is an edge-edge switch and let 
$e_1:=\hat{\sG}$, $e_2:=\hat{\tG}$.
Moreover,
let $\Ud_u$ be the ascending part of $\Ud$. 
We have the following cases:
\begin{description}
   \item[A. $e_1$ and $e_2$ are both negative in $\Ud_u$:]
   Do nothing (negative edges are not in edge-edge pairs).
   \item[B. $e_1$ is positive and $e_2$ is negative in $\Ud_u$:]
   No pairing changes by this switch.
   \item[C. $e_1$ is negative and $e_2$ is positive in $\Ud_u$:]
   Let $z$ be the representative cycle for the pair $(e_2,\epsilon)\in\Pi$.
   If $e_1\in z$, 
   pair $e_1$ with $\epsilon$ in $\Pi$
   with the same representative $z$
   (notice that $e_2$ becomes unpaired).
   \item[D. $e_1$ and $e_2$ are both positive in $\Ud_u$:]
   Let $z,z'$ be the representative cycles for the pairs $(e_1,\epsilon),(e_2,\epsilon')\in\Pi$
   respectively. Do the following according to different cases:
   \begin{itemize}
       \item If $e_1\in z'$ and the deletion of $\epsilon'$ is before the deletion of $\epsilon$ in $\Ud$:
        Let the representative for $(e_2,\epsilon')$ be $z+z'$.
        The pairing does not change.
       \item If $e_1\in z'$ and the deletion of $\epsilon'$ is after the deletion of $\epsilon$ in $\Ud$:
       Pair $e_1$ and $\eG'$ in $\Pi$
       with the representative $z'$;
       pair $e_2$ and $\eG$ in $\Pi$
       with the representative $z+z'$.
   \end{itemize}
\end{description}
\end{algr}

\begin{remark*}
Algorithm~\ref{alg:fwd-switch-ud-cc} 
has a time complexity of $O(m)$ dominated
by the summation of  1-cycles.
\end{remark*}

Proposition~\ref{prop:cc-rep-pair} and Theorem~\ref{thm:alg-fwd-switch-ud-cc-correct}
justify the correctness of Algorithm~\ref{alg:fwd-switch-ud-cc}.

\begin{proposition}\label{prop:cc-rep-pair}
For $\Lcal$ which is an up-down filtration as in Definition~\ref{def:rep},
let $E_u$ be the set of positive edges in the ascending part of $\Lcal$
and $E_d$ be the set of positive edges in the descending part of $\Lcal$.
Formalize a pairing of $E_u$ and $E_d$ as a bijection $\pi:E_u\to E_d$.
Then, the pairs $\Set{(e,\pi(e))\mid e\in E_u}$ correctly generate
$\Pers^\mathrm{cc}_1(\Lcal)$ if for each $e\in E_u$,
$(e,\pi(e))$ admits a representative cycle in $\Lcal$.
\end{proposition}
\begin{proof}
The representative defined in Definition~\ref{def:rep} is an adaption of 
the general representative for zigzag persistence 
defined in~\cite{dey2021computing,maria2014zigzag}.
Moreover, Proposition~9 in~\cite{dey2021computing} says that
if we can find a pairing for cells in $\Lcal$ s.t.\ each pair
admits a representative, then the pairing generates the barcode for $\Lcal$.
Since cells in $\Lcal$ other than those in $E_u\union E_d$
generate intervals in $\Pers_0(\Lcal)$, 
the corresponding cell pairs
which generate $\Pers_0(\Lcal)$
must admit representatives.
Combining Proposition~9 in~\cite{dey2021computing} and the assumption that each pair
$(e,\pi(e))$ admits a representative cycle,
we can arrive at the conclusion.
\end{proof}

\begin{theorem}\label{thm:alg-fwd-switch-ud-cc-correct}
Algorithm~\ref{alg:fwd-switch-ud-cc} correctly updates
the edge-edge pairs
for the switch in Equation~(\ref{eqn:fwd-switch-ud}).
\end{theorem}
\begin{proof}
By Proposition~\ref{prop:cc-rep-pair},
the correctness of the updated edge-edge pairs in Algorithm~\ref{alg:fwd-switch-ud-cc}
follows from the correctness of the updated representative cycles for these pairs.
We omit the details for verifying the validity of these cycles,
which are evident from the algorithm.
\end{proof}

As mentioned, to efficiently implement Algorithm~\ref{alg:fwd-switch-ud-cc},
we eliminate maintaining the representatives in the algorithm.
We notice that there are only two places in Algorithm~\ref{alg:fwd-switch-ud-cc}
where the pairing changes: 
\begin{itemize}
    \item Case C when $e_1\in z$;
    \item Case D when 
$e_1\in z'$ and the deletion of $\epsilon'$ is after the deletion of $\epsilon$ in $\Ud$.
\end{itemize}

For Case C,
we have the following fact:
\begin{proposition}
In Case C of Algorithm~\ref{alg:fwd-switch-ud-cc},
the edge 
$e_1$ is in $z$ 
if and only if $e_1$ is in a 1-cycle in $\hat{G}_{j+1}$ 
(refer to Equation~(\ref{eqn:fwd-switch-ud})). 
\end{proposition}
\begin{proof}
By definition,
$z$ is a 1-cycle in  $\hat{G}_{j+1}$. 
So the fact that $e_1$ is in $z$ 
trivially implies that 
$e_1$ is in a 1-cycle in $\hat{G}_{j+1}$.

To prove the `if' part of the proposition,
suppose that $e_1$ is in a cycle $z'\subseteq\hat{G}_{j+1}$
but $e_1\not\in z$.
We have that 
the 1st betti number increases by 1 
from $\hat{G}'_{j}$ to $\hat{G}_{j+1}$
because
$z'$ is a 1-cycle created by the addition of $e_1$
 in $\Ucal'$.
By definition, $e_2\in z\subseteq\hat{G}_{j+1}$.
Since $e_1\not\in z$,
we have that $z\subseteq\hat{G}'_{j}$.
So 
the 1st betti number also increases by 1 
from $\hat{G}_{j-1}$ to $\hat{G}'_{j}$
because
$z$ is a 1-cycle created by the addition of $e_2$
 in $\Ucal'$.
Therefore,
the 1st betti number  increases by 2
from $\hat{G}_{j-1}$ to $\hat{G}_{j+1}$. 
However, this is a contradiction
because the fact that $e_1$ is negative and $e_2$ is positive in $\Ud_u$
implies that
the 1st betti number increases by 1
from $\hat{G}_{j-1}$ to $\hat{G}_{j+1}$.
\end{proof}

The above proposition implies that to check whether $e_1\in z$,
we only need to check whether $e_1$ is in a 1-cycle in $\hat{G}_{j+1}$.
The checking can be similarly performed as in Algorithm~\ref{alg:std-switch-abs}
by maintaining a Link-Cut tree~\cite{sleator1981data}
for the MSF of $\hat{G}:=\hat{G}_\simpcnt$,
and then looking up the bottleneck weight of the path between $e_2$'s vertices 
in the MSF
(see Section~\ref{sec:std-switch}).
Notice that edge weights in $\hat{G}$ 
are the edges' indices in the ascending part of the up-down filtration.
This also means that to perform the backward switch on the up-down filtration,
we also need to maintain an MSF of $\hat{G}$
where the edge weights are
indices in the descending part of the up-down filtration.
Maintaining and querying the MSF is done in $O(\log m)$ time
using the Link-Cut tree~\cite{sleator1981data}.

For Case D when 
the deletion of $\epsilon'$ is after the deletion of $\epsilon$ in $\Ud$,
we only need to 
to determine whether 
to keep the pairing
(pairing $e_1$ with $\eG$ and 
$e_2$ with $\eG'$)
or
change the pairing 
(pairing $e_1$ with $\eG'$ and 
$e_2$ with $\eG$).
Let the edge $\eG'$ be deleted from $\hat{G}_h$ to form $\hat{G}_{h+1}$ in $\Ucal$,
i.e., the edge pair $(e_2,\eG')$ generates an interval $[j+1,h]\in\Pers_1(\Ucal)$.
We have the following:
\begin{observation}
\label{obsv:fwd-switch-ud-cc-caseD-change}
For Case D of Algorithm~\ref{alg:fwd-switch-ud-cc} when 
the deletion of $\epsilon'$ is after the deletion of $\epsilon$ in $\Ud$,
if there is a cycle containing $e_2$ and $\eG'$ in $\hat{G}'_j\intersect\hat{G}_h$,
then the pairing stays the same before and after the switch;
otherwise, the pairing changes.
\end{observation}
\begin{remark*}
Since we are focusing on the 1-cycles in graphs,
for the graph intersection $\hat{G}'_j\intersect\hat{G}_h$,
we only need to take the intersection of their edge sets.
For convenience, we can assume $\hat{G}'_j\intersect\hat{G}_h$ 
to contain all vertices in $\hat{G}$.
\end{remark*}

To see the correctness of the observation,
first suppose that 
there is a cycle $z''$ containing $e_2$ and $\eG'$ in $\hat{G}'_j\intersect\hat{G}_h$.
Notice that $z''$ is nothing but a representative for the pair $(e_2,\eG')$
in $\Ucal'$.
We also have that $z$ is still a representative for  the pair $(e_1,\eG)$
in $\Ucal'$.
So the paring stays the same before and after the switch 
based on Proposition~\ref{prop:cc-rep-pair}.
If there is not a cycle containing $e_2$ and $\eG'$ in $\hat{G}'_j\intersect\hat{G}_h$,
then there is no representative cycle for the pair $(e_2,\eG')$
in $\Ucal'$.
So $e_2$ cannot be paired with $\eG'$ in $\Ucal'$.

Based on Observation~\ref{obsv:fwd-switch-ud-cc-caseD-change},
in order to obtain the pairing in Case D,
we only need to check whether there is a cycle 
containing both $e_2$ and $\eG'$ in $\hat{G}'_j\intersect\hat{G}_h$.
This is equivalent to independently checking
(i) whether there is a cycle $c_1$ in $\hat{G}'_j\intersect\hat{G}_h$
containing $e_2$
and (ii) whether there is  a cycle $c_2$ in $\hat{G}'_j\intersect\hat{G}_h$
containing $\eG'$.
The reason is that, if $\eG'\in c_1$ or $e_2\in c_2$,
then either $c_1$ or $c_2$ would be 
a cycle in $\hat{G}'_j\intersect\hat{G}_h$
containing both edges;
otherwise, $c_1+c_2$ is
a cycle in $\hat{G}'_j\intersect\hat{G}_h$
containing both edges.
We then reduce the checking of (i) and (ii)
into a graph connectivity problem,
e.g., 
checking (ii) is equivalent to 
checking whether the two vertices of $\eG'$ is connected
in $\hat{G}'_{j}\intersect\hat{G}_{h+1}$.
We conclude:
\begin{observation}\label{obsv:obsv:fwd-switch-ud-cc-caseD-query}
To determine the pairing in Case D of Algorithm~\ref{alg:fwd-switch-ud-cc},
one only needs to perform two graph connectivity queries. Each query
asks whether two vertices $u,v$ are connected in $H\intersect H'$,
where $H$ is a graph in the ascending part of $\Ucal'$
and $H'$ is a graph in the descending part of $\Ucal'$.
\end{observation}

We now describe how to address the connectivity query 
as in Observation~\ref{obsv:obsv:fwd-switch-ud-cc-caseD-query}
for $\Ucal$
(the query for $\Ucal'$ can be done similarly
by first changing the data structures we maintain for $\Ucal$
into ones for $\Ucal'$).
One way to solve such a connectivity problem
for $\Ucal$
is to maintain an MSF for each $\hat{G}_\lG$ in the descending part $\Ucal_d$ of $\Ucal$ ($\lG\geq\simpcnt$),
where edges in  $\hat{G}_\lG$ are weighted by their indices in the ascending part $\Ucal_u$.
Then for a $\hat{G}_\gG$ in 
$\Ucal_u$
($\gG<\simpcnt$),
determining whether two vertices $u,v$ are connected in $\hat{G}_\gG\intersect\hat{G}_\lG$
can be done by querying the bottleneck weight of the path connecting $u,v$
in the MSF of $\hat{G}_\lG$.
Representing the MSF's as Link-Cut trees~\cite{sleator1981data},
such a query can be done in $O(\log \filtcnt)$ time.

One problem with the previous approach is that whenever we perform a forward switch on $\Ucal$,
weights for $e_1,e_2$ swap, which may induce a change on the MSF's for $O(m)$ many graphs in $\Ucal_d$.
Since updating each MSF takes $O(\log m)$ time (see Section~\ref{sec:std-switch}),
the time complexity of the update on edge-edge pairs 
then becomes $O(m\log m)$.

To perform the update faster,
we maintain MSF's for $\sqrt{m}$ many graphs in $\Ucal_d$,
where the index difference of each two consecutive graphs for which MSF's are maintained
is $\sqrt{m}$. Formally, we \emph{always} maintain MSF's for the following graphs in $\Ucal_d$:
\begin{equation}
\label{eqn:msf-graph-seq}
\hat{G}_{m}\subseteq\hat{G}_{m-\sqrt{m}}\subseteq\hat{G}_{m-2\sqrt{m}}\subseteq\cdots\subseteq\hat{G}_{m-\ell\sqrt{m}},
\end{equation}
where $\ell=\lfloor k/\sqrt{m}\rfloor$.
Notice that whenever we use `$\sqrt{m}$' in the above,
what we actually mean is either `$\lfloor\sqrt{m}\rfloor$'
or `$\lceil\sqrt{m}\rceil$'. 
We will continue abusing the notation
in the rest of this section
by using 
`$\sqrt{m}$' in place of its `floor' or `ceiling'.

Now, to obtain the MSF for an arbitrary $\hat{G}_\lG$ in $\Ucal_d$,
we first find the  graph $\hat{G}_{m-q\sqrt{m}}$ with the smallest index
which is contained in $\hat{G}_\lG$.
Then, from the MSF of $\hat{G}_{m-q\sqrt{m}}$,
we go over each edge $e$ between $\hat{G}_{m-q\sqrt{m}}$ and $\hat{G}_\lG$
in $\Ucal_d$ and see whether $e$ should be inserted into the MSF of $\hat{G}_{m-q\sqrt{m}}$.
We can get the MSF for $\hat{G}_\lG$
in $O(\sqrt{m}\,\log m)$ time using the above process
where no more than $\sqrt{m}$ edges need to be processed.
Notice that we can also understand the above process as
building the MSF for $\hat{G}_\lG$
using Kruskal's algorithm starting from the `already built'  MSF for $\hat{G}_{m-q\sqrt{m}}$,
where Link-Cut trees~\cite{sleator1981data} are used for checking the connectivity.

Notice that after making queries on the MSF of $\hat{G}_\lG$,
we need to recover the MSF of $\hat{G}_{m-q\sqrt{m}}$
so that
exactly the MSF's of graphs in Sequence~(\ref{eqn:msf-graph-seq})
are maintained.
The MSF of $\hat{G}_{m-q\sqrt{m}}$
can be recovered by recording what edges are added
to the original MSF of $\hat{G}_{m-q\sqrt{m}}$
to form the MSF of $\hat{G}_\lG$,
and then deleting those edges
from the MSF of $\hat{G}_\lG$
using the \textsc{cut} operation
in Link-Cut trees~\cite{sleator1981data}.
Since no more than $\sqrt{m}$ edges need to be
deleted,
recovering the MSF of $\hat{G}_{m-q\sqrt{m}}$
takes $O(\sqrt{m}\,\log m)$ time.

To finish a forward switch, 
we also need to update the maintained MSF's
due to the weight change on edges,
which takes $O(\sqrt{m}\,\log m)$ time
because there are no more than $\sqrt{m}$ such MSF's
and updating each MSF takes $O(\log m)$ time.
Notice that
for a backward switch, 
instead of updating $\sqrt{m}$ many MSF's maintained for $\Ucal_d$,
we need to reconstruct at most one such MSF.
The reconstruction is the same as constructing
the MSF for an arbitrary graph in $\Ucal_d$,
which takes $O(\sqrt{m}\,\log m)$ time.
Notice that to do the reconstruction,
we also need to record which edges are added to 
the MSF of $\hat{G}_{m-q\sqrt{m}}$
to form the MSF of $\hat{G}_{m-(q+1)\sqrt{m}}$,
for any consecutive graphs 
$\hat{G}_{m-q\sqrt{m}},\hat{G}_{m-(q+1)\sqrt{m}}$
in Sequence~(\ref{eqn:msf-graph-seq}).
Therefore,
updating the edge-edge pairs
takes $O(\sqrt{m}\,\log m)$ time.

We now conclude the following:
\begin{theorem}
For the switches on graph zigzag filtrations, the \emph{closed-closed} intervals in dimension 0
    can be maintained in $O(1)$ time;
    the \emph{closed-open} and \emph{open-closed} intervals, which appear only in dimension $0$,
    can be maintained in $O(\log m)$ time;
    the \emph{open-open} intervals in dimension 0 and 
    \emph{closed-closed} intervals in dimension 1
    can be maintained in $O(\sqrt{m}\,\log m)$ time. 
\end{theorem}

\paragraph{Preprocessing.}
To perform the updates for a sequence of switches starting from a graph zigzag filtration,
we need to construct the data structures maintained by the update algorithms.
Given the initial filtration,
we first compute its corresponding up-down filtration 
in $O(m)$ time~\cite{DBLP:conf/esa/DeyH22}.
For the ascending and descending parts of the initial up-down filtration,
we run the preprocessing for Algorithm~\ref{alg:std-switch-abs}
in $O(m\log m)$ time.
We then
find the oldest vertices in the ascending and descending parts 
of the initial up-down filtration
for each connected component of $\hat{G}$ in $O(m)$ time.
To find the edge-edge pairs,
we run Algorithm~\ref{alg:ext-pers} 
in $O(m\log m)$ time.
For updating the edge-edge pairs,
we also construct the MSF for each graph in Sequence~(\ref{eqn:msf-graph-seq})
individually. 
Since constructing each one is in $O(m\log m)$ time,
constructing all the MSF's is in $O(m^{1.5}\log m)$ time.
Hence, the overall preprocessing takes  $O(m^{1.5}\log m)$ time.
\bigskip

 \section{Discussion}
 We have designed $O(\log m)$ update algorithms for switches in standard graph persistence. Also, we have designed an $O(m\log m)$ algorithm for computing zigzag persistence on graphs.
 It remains open if computing graph zigzag persistence can be further improved to run in $O(m\log n)$ time.
 Similarly, for updates in graph zigzag persistence, we have designed algorithms for all 
 types of intervals. Among them,
 the updates for closed-closed intervals in dimension 1 and open-open intervals in dimension 0 take more than logarithmic time though still admitting a sub-linear algorithm ($O(\sqrt{m}\,\log m)$). 
 It would be especially interesting to see if the $O(\sqrt{m}\,\log m)$ complexity can be improved. 

  Furthermore,
 for all updates, we currently only considered switches,
  where the length of the filtrations is always fixed.
 How about those operations which change the length of the filtration,
 e.g., ones that  insert/delete some inclusion arrows
 into/from the filtration?
 Can these operations be done in logarithmic time, or at least in sub-linear time?

\bibliographystyle{plainurl}
\bibliography{refs}

\end{document}